\documentclass{article}

\usepackage{amsmath,amssymb, amsthm}
\usepackage{graphicx,enumerate}
\usepackage{ifthen}
\usepackage[usenames,dvipsnames]{color}
\usepackage{subfig}
\usepackage{algorithm}
\usepackage{algorithmic}

\newtheorem{theorem}{Theorem}
\newtheorem{lemma}[theorem]{Lemma}

\newtheorem{definition}{Definition}
\newtheorem{example}{Example}

\begin{document}
\title{Peer Effects and Stability in Matching Markets}
\author{Elizabeth Bodine-Baron \and Christina Lee \and Anthony Chong \and Babak Hassibi \and Adam Wierman\\
California Institute of Technology\\ Pasadena, CA 91125, USA\\
\{eabodine, chlee, anthony, hassibi, adamw\}@caltech.edu \thanks{This work was supported in part by the National Science Foundation under grants CCF-0729203, CNS-0932428 and CCF-1018927, by the Office of Naval Research under the MURI grant N00014-08-1-0747, and by Caltech's Lee Center for Advanced Networking.}
} 
\maketitle

\begin{abstract}

Many-to-one matching markets exist in numerous different forms, such as college admissions, matching medical interns to hospitals for residencies, assigning housing to college students, and the classic firms and workers market.  In all these markets, externalities such as complementarities and peer effects severely complicate the preference ordering of each agent.  Further, research has shown that externalities lead to serious problems for market stability and for developing efficient algorithms to find stable matchings.  In this paper we make the observation that peer effects are often the result of underlying social connections, and we explore a formulation of the many-to-one matching market where peer effects are derived from an underlying social network. The key feature of our model is that it captures peer effects and complementarities using utility functions, rather than traditional preference ordering. With this model and considering a weaker notion of stability, namely two-sided exchange stability, we prove that stable matchings always exist and characterize the set of stable matchings in terms of social welfare.  We also give distributed algorithms that are guaranteed to converge to a two-sided exchange stable matching.  To assess the competitive ratio of  these algorithms and to more generally characterize the efficiency of matching markets with externalities, we provide general bounds on how far the welfare of the worst-case stable matching can be from the welfare of the optimal matching, and find that the structure of the social network (e.g. how well clustered the network is) plays a large role.  
\end{abstract}

\section{Introduction}

Many-to-one matching markets exist in numerous forms, such as college admissions, the national medical residency program, freshman housing assignment, as well as the classic firms-and-workers market.  These markets are widely studied in academia and also widely deployed in practice, and have been applied to other areas, such as FCC spectrum allocation and supply chain networks \cite{BajariFox09,Ostrovsky08}

In the conventional formulation, matching markets consist of two sets of agents, such as medical interns and hospitals, each of which have preferences over the agents to which they are matched.  In such settings it is important that matchings are `stable' in the sense that agents do not have incentive to change assignments after being matched.  The seminal paper on matching markets was by Gale and Shapley \cite{GaleShapley62}, and following this work an enormous literature has grown, e.g., \cite{KojimaPathak09,RothRothblumVandeVate93,matching_book,RothVandeVate90} and the references therein.  Further, variations on Gale and Shapley's original algorithm for finding a stable matching are in use by the National Resident Matching Program (NRMP), which matches medical school graduates to residency positions at hospitals \cite{Roth84}.

However, there are problems with many of the applications of matching markets in practice.  For example, couples participating in the NRMP often reject their matches and search outside the system.   In housing assignment markets where college students are asked to list their preferences over housing options, there  is often collusion among friends to list the same preference order for houses.  These two examples highlight that `peer effects', whether just couples or a more general set of friends, often play a significant role in many-to-one matchings.  That is, agents care not only where they are matched, but also which other agents are matched to the same place.  Similarly, `complementarities' often play a role on the other side of the market.  For example, hospitals and colleges care not only about which individual students are assigned to them, but also that the group has a certain diversity, e.g., of different specializations.

As a result of the issues highlighted above, there is a growing literature studying many-to-one matchings with externalities (i.e., peer effects and complementarities) \cite{Dutta97,Hafalir08,KlausKlijn05,KlausKlijn07,Pycia07,Revilla04,YarivWilson09,EcheniqueYenmez07,SasakiToda96} and the research has found that designing matching mechanisms is significantly more challenging when externalities are considered, e.g. incentive compatible mechanism design is no longer possible.  In fact, even determining the existence of stable matchings in the presence of externalities has been difficult.

The reason for the difficulty is that there is no longer a guarantee that a stable many-to-one matching will exist when agents care about more than their own matching \cite{Roth84,matching_book}, and, if a stable matching does exist, it can be computationally difficult to find \cite{Ronn90}. Consequently, most research has focused on identifying when stable matchings do and do not exist.  Papers have proceeded by constraining the matching problem through restrictions of the possible preference orderings, \cite{Dutta97,Hafalir08,KlausKlijn05,KlausKlijn07,Pycia07,Revilla04}, and by considered variations on the standard notion of stability \cite{YarivWilson09,EcheniqueYenmez07,SasakiToda96}.  Our paper also considers a restriction of the model, described in the following.

The key idea of this paper is that \emph{peer effects are often the result of an underlying social network}.  That is, when agents care about where other agents are matched, it is often because they are friends.  With this in mind, we construct a model in Section \ref{section:model} that includes a weighted, undirected social network graph and allows agents to have utilities (which implicitly defines their preference ordering) that depend on where neighbors in the graph are assigned.  The model is motivated by \cite{YarivWilson09}, which also considers peer effects defined by a social network but focuses on one-sided matching markets rather than two-sided matching markets.  

We focus on \emph{two-sided exchange-stable} matchings -- see Section \ref{section:model} for a detailed definition.  We note that compared to the traditional notion of stability of \cite{GaleShapley62}, this is a distinct notion of stability, but one that is relevant to many situations where agents can compare notes with each other, such as the housing assignment or medical matching problem.  For example, in \cite{YarivWilson09,BajariFox09,Fox09}, ``pairwise-stability'' is considered since they consider models where agents exchange offices or licenses in FCC spectrum auctions.  Further, consider a situation where two hospital interns prefer to exchange the hospitals allocated to them by the NRMP.  If this is a traditional stable matching, the hospitals would not allow the swap, even though the interns are highly unsatisfied with the match.  Such a situation has been documented in \cite{Irving08}, and has led to a similar type of stability, exchange stability, as defined in \cite{Alcalde94,Cechlarova02,Cechlarova05,Irving08}.

Given our model of peer effects, the focus of the paper is then on characterizing the set of two-sided exchange-stable matchings.  Our results concern (i) the existence of two-sided exchange-stable matchings, (ii) algorithms for finding two-sided exchange-stable matchings, and (iii) the efficiency of exchange-stable matchings (in terms of social welfare).

With respect to the existence of stable matchings (Section \ref{sec:existence}), it is not difficult to show that in our model two-sided exchange-stable matchings always exist. Further, if students value houses according to the same rules, the matching that maximizes social welfare is guaranteed to be two-sided exchange-stable.  Given the contrast to the negative results that are common for many-to-one matchings, e.g., \cite{EcheniqueYenmez07,Ronn90,Roth84}, these results are perhaps surprising. Further, the results on characterizing the existence of stable matchings naturally suggest two simple algorithms for finding stable matchings, which we discuss in Section \ref{sec:sim}.

With respect to the efficiency of exchange-stable matchings (Section \ref{sec:PoA}), results are not as easy to obtain.  In this context, we limit our focus to one-sided matching markets, but as a result we are able to attain bounds on the ratio of the welfare of the optimal matching to that of the worst stable matching, i.e., the `price of anarchy'.   We also demonstrate cases where our bounds are tight.  When considering only one-sided markets, our model becomes similar to hedonic coalition formation, but with several key differences, as highlighted in Section \ref{sec:PoA}.  Our results (Theorems \ref{thm:PoASimple} and \ref{thm:PoAGeneral}) show that the price of anarchy does not depend on the number of, say, interns, but does grow with the number of, say, hospitals -- though the growth is typically sublinear.  Further, we observe that the impact of the structure of the social network on the price of anarchy happens only through the clustering of the network, which is well understood in the context of social networks, e.g., \cite{Jackson08,WassermanFaust94}.  Finally, it turns out that the price of anarchy has a dual interpretation in our context; in addition to providing a bound on the inefficiency caused by enforcing exchange-stability, it turns out to also provide a bound on the loss of efficiency due to peer effects. 
\section{Model and notation} \label{section:model}

To begin, we define the model we use to study many-to-one matchings with peer effects and complementarities.  There are four components to the model, which we describe in turn: (i) basic notation for discussing matchings; (ii) the model for agent utilities, which captures both peer effects and complementarities; (iii) the notion of stability we consider; and (iv) the notion of social welfare we consider.

To provide a consistent language for discussing many-to-one matchings, throughout this paper we use the setting of matching incoming undergraduate students to residential houses.  In this setting many students are matched to each house, and the students have preferences over the houses, but also have peer effects as a result of wanting to be matched to the same house as their friends.  Similarly, the houses have preferences over the students, but there are additional complementarities due to goals such as maintaining diversity.  It is clear that some form of stability is a key goal of this ``housing assignment'' problem.

\paragraph{Notation for many-to-one matchings}
Using the language of the housing assignment problem, we define two finite and disjoint sets, $H=\{h_1, \dots, h_m\}$ and $S=\{s_1, \dots, s_n\}$ denoting the houses and students, respectively.  For each house, there exists a positive integer \emph{quota} $q_h$ which indicates the number of positions a house has to offer.  The quota for each house may be different.

A matching $\mu$ describes the assignment of students to houses such that students are matched to only one house, while houses are matched to multiple students. More formally:
\begin{definition} A \emph{matching} is a subset $\mu \subseteq S \times H$ such that $|\mu(s)| = 1$ and $|\mu(h)| = q_h$, where $\mu(s) = \{ h \in H: (s,h) \in \mu\}$ and $\mu(h) = \{ s \in S: (s,h) \in \mu\}$.\footnote{If the number of students in $\mu(h)$, say $r$, is less than $q_h$, then $\mu(h)$ contains $q_h - r$ ``holes'' -- represented as students with no friends and no preference over houses.}
\end{definition}
Note that we use $\mu^{2}(s)$ to denote the set of student $s$'s housemates (students also in house $\mu(s)$).

\paragraph{Friendship network}
The friendship network among the students is modeled by a weighted graph, $G=(V,E,w)$ where $V=S$ and the relationships between students are represented by the weights of the edges connecting nodes.  The strength of a relationship between two students $s$ and $t$ is represented by the weight of that edge, denoted by $w(s,t) \in \mathbb{R}^+\cup\{0\}$. We require that the graph is undirected, i.e., the adjacency matrix is symmetric so that $w(s,t) = w(t,s)$ for all $s,t$.

Additionally, we define a few metrics quantifying the graph structure and its role in the matching.  Let the total weight of the graph be denoted by $|E| :=  \frac{1}{2} \sum_{s \in S} \sum_{t \in S} w(s,t)$.  Further, let the weight of edges connecting houses $h$ and $g$ under matching $\mu$ be denoted by  $E_{hg}(\mu) := \sum_{s \in \mu(h)} \sum_{t \in \mu(g)} w(s,t)$. Note that in the case of edges within the same house $E_{hh}(\mu) := \frac{1}{2}\sum_{s \in \mu(h)} \sum_{t \in \mu(h)} w(s,t)$. Finally, let the weight of edges that are within the houses of a particular matching $\mu$ be denoted by $E_{in}(\mu) := \sum_{h\in H} E_{hh}(\mu).$

\paragraph{Agent utility functions}
In our model, each agent derives some utility from a particular matching and an agent (student or house) always strictly prefers matchings that give a strictly higher utility and is indifferent between matchings that give equal utility. This setup differs from the traditional notion of `preference orderings' \cite{GaleShapley62,matching_book},  but is not uncommon \cite{AnsheDas09,YarivWilson09,BajariFox09,BranzeiLarson09,Fox09}.   It is through the definitions of the utility functions that we model peer effects (for students) and complementarities (for houses).

Under our model, students derive benefit both from (i) the house they are assigned to and (ii) their peers that are assigned to the same house. We model each house $h$ as having an desirability of $D^s_h \in \mathbb{R}^+\cup\{0\}$ for student $s$. A similar model was first used in \cite{YarivWilson09} and is meant to capture the physical characteristics of the house (amenities, size, etc.) with $D^s_h$ as well as peer effects.  If $D^s_h = D^t_h \; \forall s\neq t$ (objective desirability), this value can be seen as representing something like the U.S. News college rankings or hospital rankings -- something that all students would agree on.  This leads to a utility for student $s$ under matching $\mu$ of
\begin{equation}
U_s(\mu) := D^s_{\mu(s)} + \sum_{t\in\mu^2(s)}w(s,t)
\end{equation}
where $w(s,t)$ is the weight of the edge between $s$ and $t$ in the friendship graph and $D^s_{h}$ is utility derived by student $s$ for house $h$, so that the total utility that a student derives from a match is a combination of how ``good'' a house is as well as how many friends they will have in that house.\footnote{We note that the utility of any ``holes'' (such as what happens when a house's quota is not met), is simply $U_s(\mu) = 0$.}\footnote{Note also that if we remove $D_h$ from the utility function and allow unlimited quotas, the matching problem becomes the coalitional affinity game from \cite{BranzeiLarson09}.}   

Similarly, the utility of a house $h$ under matching $\mu$ is modeled by
\begin{equation} U_h(\mu) := D_{\mu(h)}^h, \end{equation}
where $D_{\sigma}^h$ denotes the desirability of a particular set of students $\sigma$ for house $h$ (the utility house $h$ derives from being matched to the set of students $\sigma$).  Note that this definition of utility allows general phenomena such as heterogeneous house preferences over groups of students.

\paragraph{Two-sided exchange stability}
Under the traditional definition of stability, if a student and a house were to prefer each other to their current match (forming a blocking pair), the student is free to move to the preferred house and the house is free to evict (if necessary) another student to make space for the preferred student.  In our model, however, we assume that students and houses cannot ``go outside the system''  and leave the university (neither can students remain unmatched), like what medical students and hospitals do when they operate outside of the NRMP.  As a result, we restrict ourselves to considering swaps of students between houses, similar to \cite{YarivWilson09,BajariFox09,Fox09}.

To define exchange stability, it is convenient to first define a \emph{swap matching} $\mu_s^{t}$ in which students $s$ and $t$ switch places while keeping all other students' assignments the same.
\begin{definition}
A \textbf{swap matching} $\mu_s^t = \{\mu \setminus \{(s,h),(t,g)\} \} \cup \{(s,g),(t,h)\}.$
\end{definition}

Note that the agents directly involved in the swap are the two students switching places and their respective houses -- all other matchings remain the same.  Further, one of the students involved in the swap can be a ``hole'' representing an open spot, thus allowing for single students moving to available vacancies.  When two actual students are involved, this type of swap is a two-sided version of the ``exchange'' considered in \cite{Alcalde94,Cechlarova02,Cechlarova05,Irving08} -- \emph{two-sided} exchange stability requires that houses approve the swap.  As a result, while an exchange-stable matching may not exist in either the marriage or roommate problem, we show in Section \ref{sec:existence} that a two-sided exchange-stable matching will always exist for the housing assignment problem.

\begin{definition} \label{psDef_orig} A matching $\mu$ is \textbf{two-sided exchange-stable (2ES)} if and only if there does not exist a pair of students $(s,t)$ such that:
\begin{enumerate}[(i)]
\item $\forall\;i \in \{s,t,\mu(s),\mu(t)\}$, $U_i(\mu_s^t)\geq U_i(\mu)$ and
\item $\exists\;i \in \{s,t,\mu(s),\mu(t)\}$ such that $U_i(\mu_s^t) > U_i(\mu)$
\end{enumerate}
\end{definition}

This definition implies that a swap matching in which all agents involved are indifferent is two-sided exchange-stable.  This avoids looping between equivalent matchings.  Note that the above definition implies that if two students want to switch between two houses (or a single student wants to ``switch'' with a hole), the houses involved must ``approve'' the swap or if two houses want to switch two students, the students involved must agree to the swap (a hole will always be indifferent).  This is natural for the house assignment problem and many other many-to-one matching markets, but would be less appropriate for some other settings, such as the college-admissions model.

\paragraph{Social welfare}
One key focus of this paper is to develop an understanding of the ``efficiency loss'' that results from enforcing stability of assignments in matching markets.  We measure the efficiency loss in terms of the ``social welfare'', which we define as follows:
\begin{equation*} 
W(\mu) := \sum_{s\in S}U_s(\mu) + \sum_{h\in H}U_h(\mu)
\end{equation*}

Using this definition of social welfare, the efficiency loss can be quantified using the \emph{Price of Anarchy} (PoA) and \emph{Price of Stability} (PoS).  Specifically, the PoA (PoS) is the ratio of the optimal social welfare over all matchings, not necessarily stable, to the minimum (maximum) social welfare over all stable matchings.  Understanding the PoA and PoS is the focus of Section \ref{sec:PoA}.
\section{Existence of stable matchings } \label{sec:existence}

We begin by focusing on the existence of two-sided exchange stable matchings.   In most prior work, matching markets with externalities do not have guaranteed existence of a stable matching.  For example, in the presence of couples on the resident side of the hospital matching market, the NRMP algorithm may fail to have a stable outcome \cite{Roth84,matching_book}, and even if a stable matching does exist, it may be NP-hard to find \cite{Ronn90}.

In contrast to the prior literature discussed above, we prove that a two-sided exchange stable matching always exists in the model considered in this paper. We begin by proposing a potential function $\Phi(\mu)$ for the matching game:\begin{equation} \label{potFn}
\Phi(\mu) = \sum_{h \in H} U_h(\mu) + \sum_{s \in S} D^s_{\mu(s)} + \frac{1}{2} \sum_{s \in S} \left( \sum_{x \in \mu^2(s)} w(s,x) \right) \end{equation}

Due to the symmetry of the social network, every approved swap will result in a strict increase of the potential function. Specifically, 
\begin{lemma} \label{lem:existencehelper_general} Any swap matching $\mu_s^t$ for which
\begin{enumerate}[(i)]
\item $\forall \: i \in \{s,t,\mu(s),\mu(t)\}$, $U_i(\mu_s^t)\geq U_i(\mu)$, and
\item $\exists \: i \in \{s,t,\mu(s),\mu(t)\}$ with $U_i(\mu_s^t) > U_i(\mu)$
\end{enumerate}
has $\Phi(\mu_s^t) > \Phi(\mu)$.  \end{lemma}
Detailed proofs of this lemma and the following results are shown in Appendix \ref{app:existence}.  Expanding on this idea, it is easy to prove the following theorem.
\begin{theorem} \label{thm:stable_general} All local maxima of $\Phi(\mu)$ are two-sided exchange stable. \end{theorem}
As there is a finite set of matches, this results in the existence of a two-sided exchange stable matching for every housing assignment market. 

If we assume that there are no vacancies in any of the houses and students value houses according to the same rules (i.e., $D_h^s = D_h^t \; \forall \; s \neq t$), then each each approved swap will result in a strict increase in the \emph{social welfare}.  Specifically,
\begin{lemma} \label{lem:existencehelper} If house quotas are exactly met and $D_h^s = D_h^t \; \forall \; s \neq t$, then any swap matching $\mu_s^t$ for which
\begin{enumerate}[(i)]
\item $\forall \: i \in \{s,t,\mu(s),\mu(t)\}$, $U_i(\mu_s^t)\geq U_i(\mu)$, and
\item $\exists \: i \in \{s,t,\mu(s),\mu(t)\}$ with $U_i(\mu_s^t) > U_i(\mu)$
\end{enumerate}
has $W(\mu_s^t) > W(\mu)$.  \end{lemma}
As before, it is now easy to prove the following theorem.
\begin{theorem} \label{thm:stable_SW} If house quotas are exactly met and $D_h^s = D_h^t \; \forall \; s \neq t$, all local maxima of $W(\mu)$ are two-sided exchange stable.  \end{theorem} 
Note that this implies that the maximally efficient matching will be two-sided exchange-stable.\footnote{Note that a local maximum of $W(\mu)$ is a matching $\mu$ for which there exists no matching $\mu'$ which is obtained from $\mu$ by swapping the assignment of exactly two students (or a student and an empty spot) and has $W(\mu')>W(\mu)$.}  

Note, however, that not all two-sided exchange-stable matchings are local maxima of $\Phi(\mu)$ or $W(\mu)$.  Such a case arises when one student, for example, refuses a swap as her utility would decrease, but the other student involved stands to benefit a great deal from such a swap.  If the swap were forced, the total potential function (or social welfare) could increase, but only at the expense of the first student.

The contrast between Theorem \ref{thm:stable_general} and the results such as \cite{Roth84} and \cite{matching_book} can be explained by considering a few aspects of the model we study.  In particular, we are using a distinct type of stability appropriate to our housing assignment market.  Further,  the assumption that the social network graph is symmetric are key to guaranteeing existence.     
\section{Finding stable matchings} \label{sec:sim}

In the previous section we have shown that a two-sided exchange-stable matching will always exist and, moreover, that under certain assumptions, socially optimal matchings are two-sided exchange-stable.  In this section, we turn to the task of developing algorithms for finding two-sided exchange-stable matchings. In particular, two natural algorithms follow immediately from our analysis.  For simplicity, in this section we will assume the conditions of Theorem \ref{thm:stable_SW}; namely, that quotas are exactly met and students rate houses according to the same rules.\footnote{We note that the results of this section extend to the more general case, using the potential function defined in Section \ref{sec:existence} rather than the social welfare function.}

\begin{algorithm} \caption{(Greedy)} \label{alg:greedy}
\begin{algorithmic}
\WHILE {$i \leq$ maxIterations}
	\STATE Search for ``approved'' swap $\mu_s^t$
	\STATE $\mu \gets \mu_s^t$
	\STATE $i \gets i+1$
\ENDWHILE
\end{algorithmic}
\end{algorithm}

Algorithm \ref{alg:greedy} proceeds by greedily considering ``approved'' swaps among students/houses that improve the social welfare.  Note that this algorithm can easily be implemented in a distributed manner, and loosely models the process by which individuals adjust a matching that is not stable.

Lemma \ref{lem:existencehelper} and Theorem \ref{thm:stable_general} immediately give that this algorithm will converge to a two-sided exchange-stable matching, since the social welfare strictly improves with each iteration, and all local maxima of $W$ are two-sided exchange-stable matchings.  Note that Algorithm \ref{alg:greedy} is not guaranteed to converge to the socially optimal stable matching; it will likely find just a local maxima of $W$.  Also, note that each iteration of the algorithm above can involve searching many pairs of students (and houses) for an ``approved'' swap.

The second algorithm we consider again seeks to optimize $W$, this time using a MCMC heat bath. In this algorithm we start with a random initial matching and at each iteration swap a random pair of students with a probability that depends on the change in social welfare: a positive change yields a probability of swapping larger than 1/2 and vice--versa.  Algorithm \ref{alg:MCMC} therefore can emerge from a local maximum.  The algorithm keeps track of the ``best'' matching found so far, even as it moves to worse matchings.  If Algorithm \ref{alg:MCMC} is run sufficiently long (perhaps exponential time) it can find the optimal two-sided exchange-stable matching \cite{haggstrom}.  However, there is no guarantee that the best matching encountered in finite time is even two-sided exchange-stable, a situation that can be remedied by applying the greedy algorithm to this matching.  Simulation results show the superiority of Algorithm \ref{alg:MCMC} to Algorithm \ref{alg:greedy} in terms of welfare, at the expense of an increase in the number of computations.


\begin{algorithm} [h]\caption{MCMC} \label{alg:MCMC}
\begin{algorithmic}
\WHILE {$i \leq$ maxIterations}
	\STATE Pick random pair of students $\{s, t\}$
	\STATE $P_T = \frac{1}{1+e^{-T(W(\mu_s^t) - W(\mu))}}$
	\STATE $\mu \gets \mu_s^t$ with probability $P_T$
	\IF {$(W(\mu_s^t) > W_{best})$}
		\STATE $W_{best} = W(\mu_s^t)$
	\ENDIF
	\STATE $i \gets i+1$
\ENDWHILE
\end{algorithmic} \end{algorithm}

To illustrate the performance of these two natural algorithms, we use two social network data sets.

The first data set is from the Caltech Project \cite{Ensminger}.  This data set is an undirected graph representing the friendship links among the undergraduates at the California Institute of Technology in 2010.  It includes approximately 900 nodes and 3500 edges.  To illustrate the algorithms, we created 10 houses and assigned them desirability values uniformly distributed from 0 to 10.  For the other side of the market, each student is assigned a score by each house, uniformly distributed from 0 to 10.  Each algorithm is run using the same initial values, and the results are shown in Figure \ref{fig:algtest}.

The second data set we use is from voting records for admin promotion at Wikipedia.  Edges in the dataset represent votes for or against a user by another user.  For simplicity, we treated the directed graph as undirected, resulting in approximately 7000 nodes and 100000 edges. To illustrate the algorithms, we created 71 houses and assigned them uniform random values between 0 and 10 as before. For the other side of the market, each users is assigned a score by each house, uniformly distributed from 0 to 10, also as before.  Each algorithm is run using the same initial values, and the results are shown in Figure \ref{fig:algtest2}.

Figures \ref{fig:algtest} and \ref{fig:algtest2} illustrate that with both networks, Algorithm 2 has longer running time than Algorithm \ref{alg:greedy}, which converges quickly.\footnote{We note that in the greedy algorithm, an ``iteration'' can take much more time to complete than one ``iteration'' of the MCMC heat bath.  Even with this effect, however, the MCMC takes longer than the greedy algorithm.} The y-axis in all figures shows the social welfare of the matching at each iteration.  As expected, Algorithm \ref{alg:greedy} converges to a sub-optimal matching for both networks, but this value is of the same order of magnitude as that found by Algorithm \ref{alg:MCMC}.

\begin{figure} [t]  \centering
  \subfloat[Algorithm 1]{\label{fig:greedyCaltech}\includegraphics[width=0.45\textwidth]{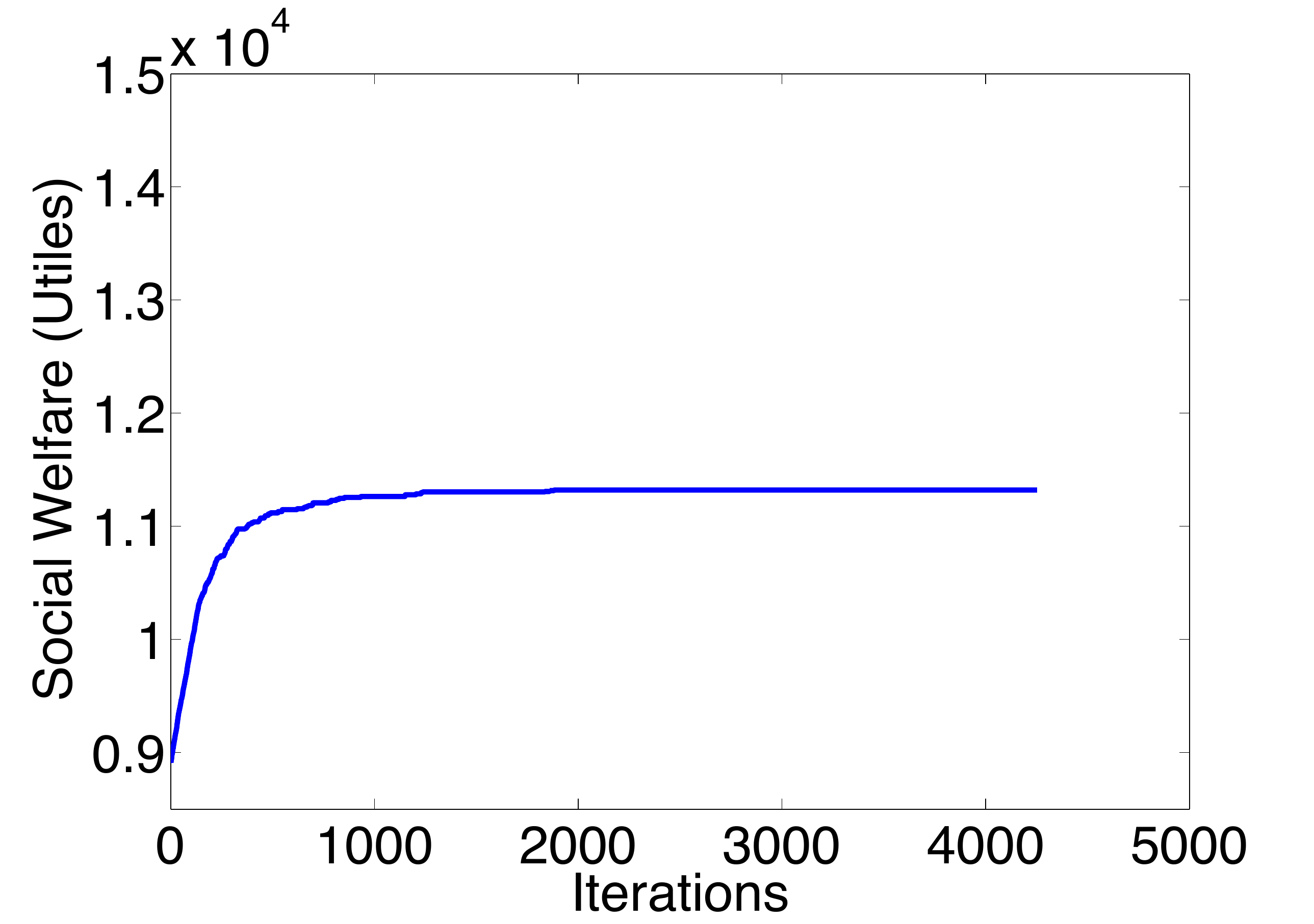}}
  \subfloat[Algorithm 2]{\label{fig:heatCaltech}\includegraphics[width=0.45\textwidth]{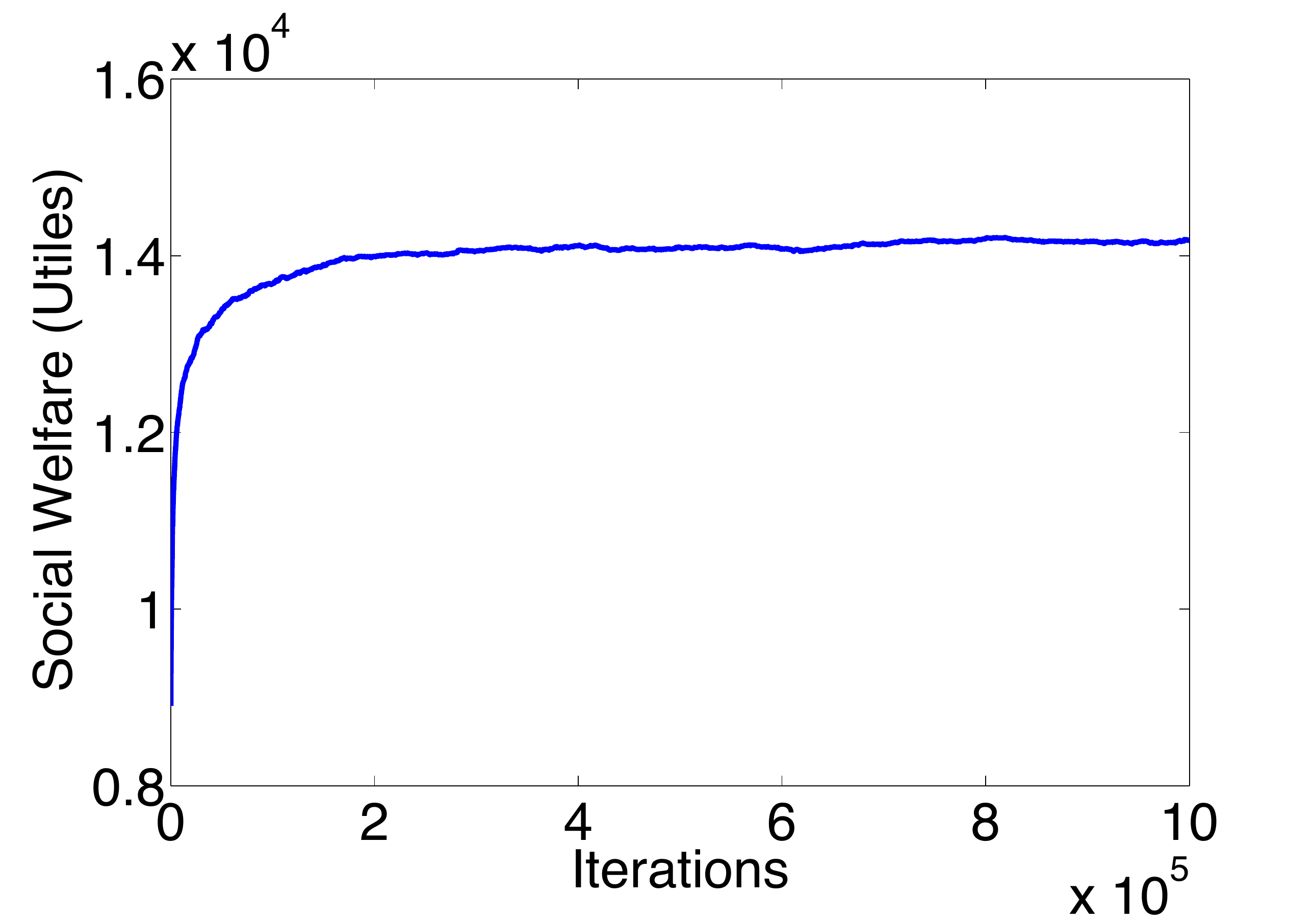}}
  \caption{Illustration of the performance of Algorithms \ref{alg:greedy} and \ref{alg:MCMC} on the Caltech Social Network}
  \label{fig:algtest}
\end{figure}

\begin{figure} [t]  \centering
  \subfloat[Algorithm 1]{\label{fig:greedyWiki}\includegraphics[width=0.45\textwidth]{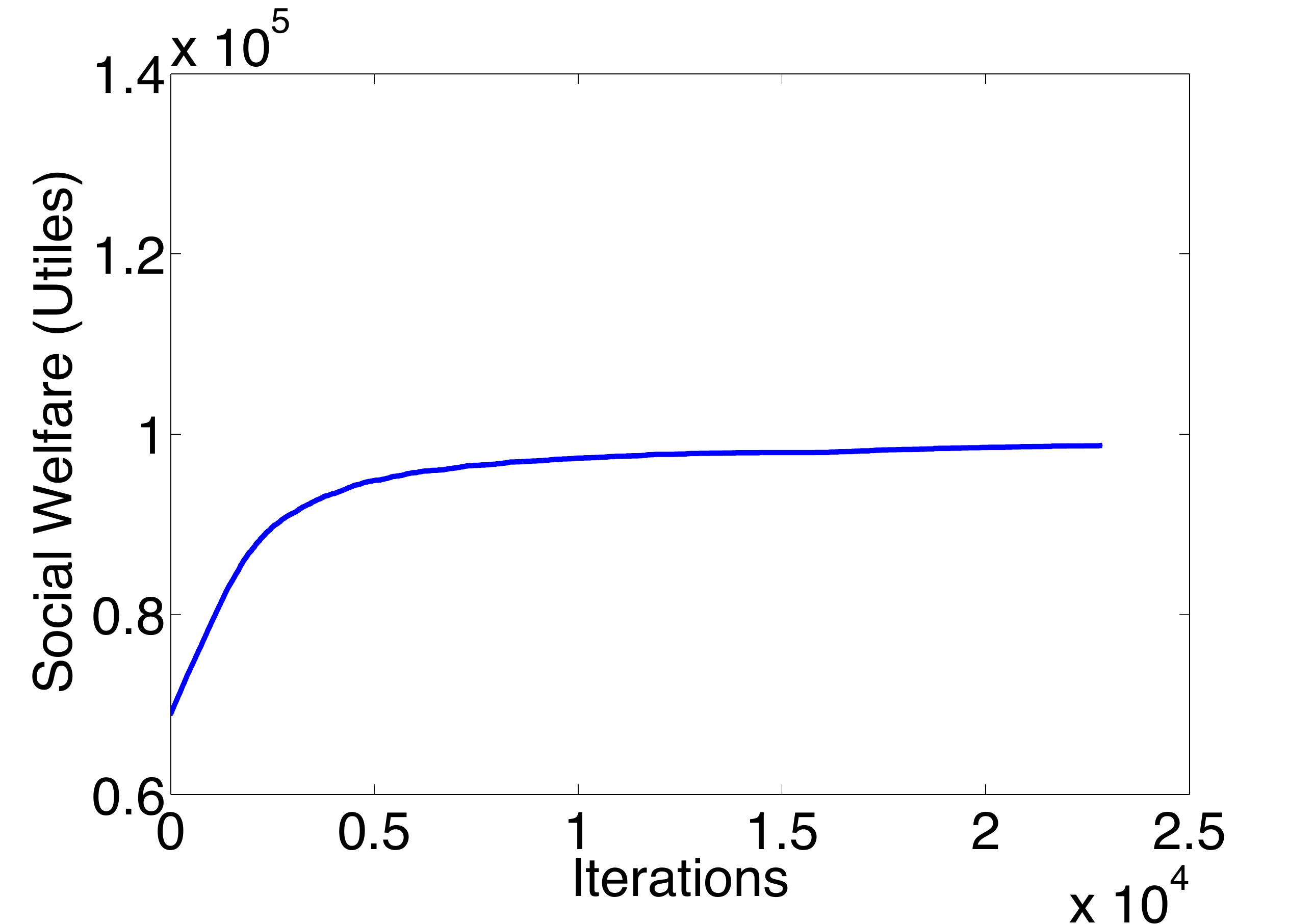}}
  \subfloat[Algorithm 2]{\label{fig:heatWiki}\includegraphics[width=0.45\textwidth]{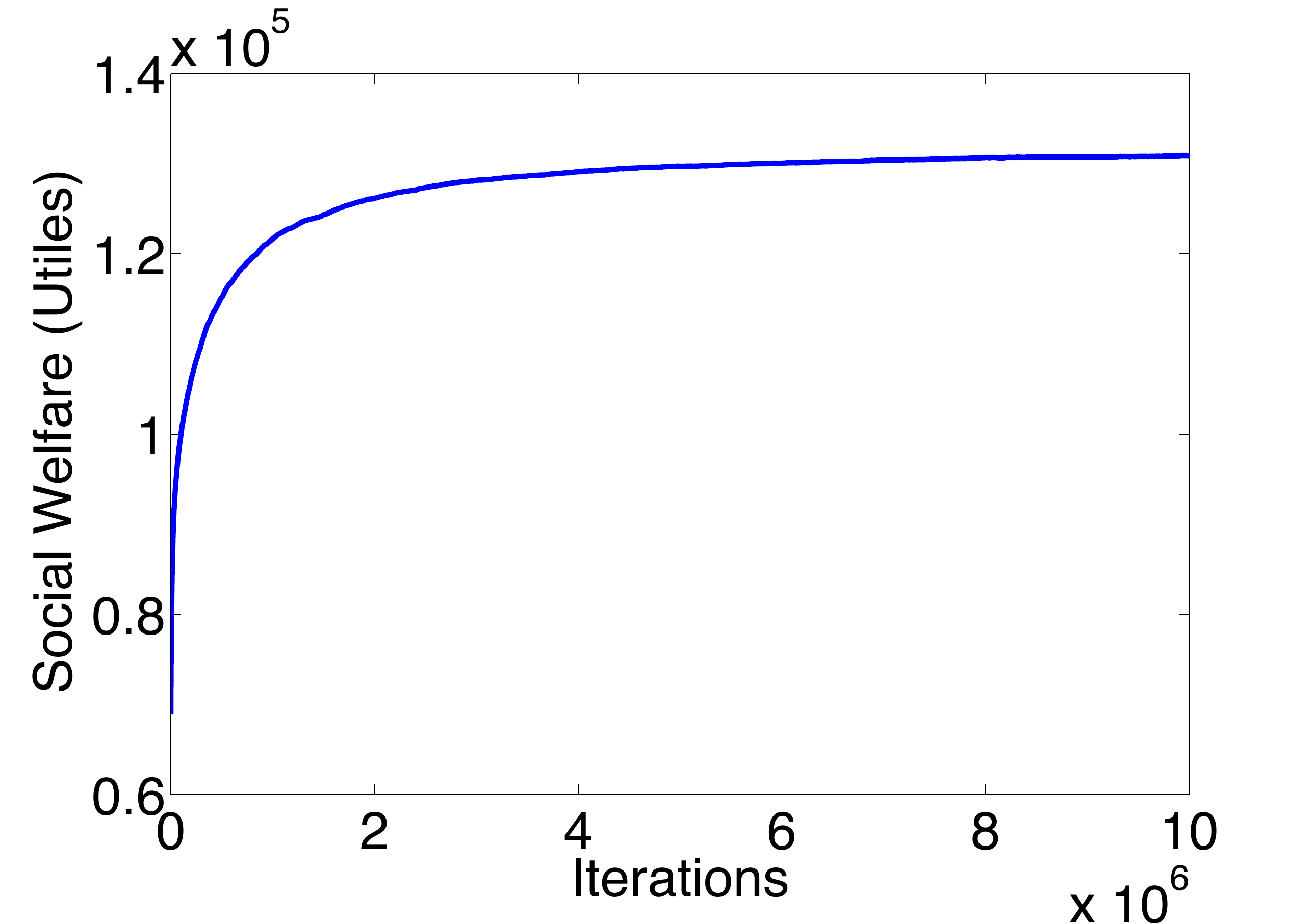}}
  \caption{Illustration of the performance of Algorithms \ref{alg:greedy} and \ref{alg:MCMC} on the Wikipedia Voting network}
  \label{fig:algtest2}
\end{figure}
\section{Efficiency of stable matchings}\label{sec:PoA}

To this point, we have focused on the existence of two-sided exchange-stable matchings and how to find them.  In this section our focus is on the ``efficiency loss''  due to stability in a matching market and the role peer effects play in this efficiency loss.

We measure the efficiency loss in a matching market using the price of stability (PoS) and the price of anarchy (PoA) as defined in Section \ref{section:model}.  Interestingly, the price of anarchy has multiple interpretations in the context of this paper.  First, as is standard, it measures the worst-case loss of social welfare that results due to enforcing exchange-stability.  For example, the authors in \cite{AnsheDas09} bound the loss in social welfare caused by individual rationality (by enforcing stable matchings) for matching markets without externalities.  Second, it provides a competitive ratio for Algorithm \ref{alg:greedy} for finding a stable matching, since Algorithm 1 provides no guarantee about which stable matching it will find.  Third, we show later that the price of anarchy also has an interpretation as capturing the efficiency lost due to peer effects.

The results in this section all require one additional simplifying assumption to our model: \emph{complementarities are ignored and only peer effects are considered}.  Specifically, we assume, for all of our PoA results, $U_h(\mu) = 0$, and thus $W(\mu)=\sum_{s\in S} U_s(\mu)$.  Under this assumption, the market is one-sided, with only students participating-- as a result we are only considering exchange-stability.  This assumption is limiting, but there are still many settings within which the model is appropriate.  Two examples are the housing assignment problem in the case when students can swap positions without needing house approval, and the assignment of faculty to offices as discussed in \cite{YarivWilson09}, as clearly the offices have no preferences over which faculty occupy them.  In order to simplify the analysis, we also make use of the assumptions in Theorem \ref{thm:stable_SW} : (i) $D^s_h = D^t_h \; \forall \; s \neq t$ and (ii) house quotas are exactly met.

\subsection{Related models}

When the housing assignment problem is restricted to a one-sided market involving only students, we note that it becomes very similar to both (i) a hedonic coalition formation game with symmetric additively separable preferences, as described in \cite{JacksonBogo02}, and (ii) a coalitional affinity game, as described in \cite{BranzeiLarson09}.  

In hedonic coalition formation games, agents' preferences for a given coalition are based on the other members of that coalition \cite{DrezeGreenberg80}. Note that coalition games are necessarily one-sided -- agents care about the coalitions but the coalitions cannot care about the agents.  The most related work to ours in this area is \cite{JacksonBogo02}, where the authors show that when agents' preferences over coalition are symmetric and additively separable (as the student utility functions in the housing assignment problem are), a Nash (and individually) stable coalition structure will always exist.  This mimics the existence result proved in Section \ref{sec:existence}, however our result applies for a two-sided market. Further work on hedonic games looks at the complexity of finding stable coalition structures; see \cite{Aziz10,BuraniZwicker03,ElkindWooldridge09,GairingSavani10,Olsen09} for examples.

Coalitional affinity games consider the pairwise relationships between agents, as represented by a weighted graph \cite{BranzeiLarson09}, and are a special subclass of hedonic games.  The most related result to the current work is \cite{BranzeiLarson09}, which proves a tight upper bound on the Price of Anarchy using the notion of core stability\footnote{A coalition structure is core stable if no set of agents can break away and form a new coalition to improve their own utility.} when the weighted graph is symmetric.

While the one-sided housing assignment problem and hedonic coalition formation games appear to be very similar, there are a number of key differences.  Most importantly, the housing assignment problem considers a \emph{fixed} number of houses with a limited number of spots available; students cannot break away and form a new coalition/house, nor can a house have more students than its quota.  In addition, our model considers exchange-stability, which is closest to the Nash stability of \cite{JacksonBogo02}, but is still significantly different in that it involves a pair of students willing and able to swap.  Finally, each student gains utility from the house they are matched with, in addition to the other members of that house, which is different from the original formulation of hedonic coalition games.

\subsection{Discussion of results}

To begin the discussion of our results, note that, as discussed in Section \ref{sec:existence}, the price of stability is 1 for our model because any social welfare optimizing matching is stable.

However, the price of anarchy can be much larger than 1.  In fact, depending on the social network, the price of anarchy can be unboundedly large, as illustrated in the following example.

\begin{figure} [t] \centering 
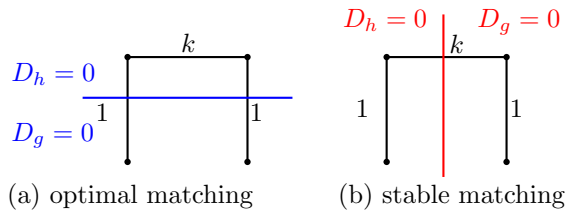
\caption{Arbitrarily bad exchange-stable matching} \label{fig:arbBad} \end{figure}

\begin{example}[Unbounded price of anarchy]

Consider a matching market with 4 students and 2 houses, each with a quota of 2, and two possible matchings illustrated by Figure \ref{fig:arbBad}.  As shown in Figure \ref{fig:arbBad} (a) and (b), respectively, in the optimal matching $\mu^*$, $W(\mu^*) = k$; whereas there exists a exchange-stable matching with $W(\mu) = 2$.  Thus, as $k$ increases, the price of anarchy grows linearly in $k$.
\end{example}

Despite the fact that, in general, there is a large efficiency loss that results from enforcing exchange-stability, in many realistic cases the efficiency loss is actually quite small. The following two theorems provide insight into such cases.

A key parameter in these theorems is $\gamma_m^*$ which captures how well the social network can be ``clustered'' into a \emph{fixed} number of $m$ groups and is defined as follows.
\begin{align}
\label{gammaDef}\gamma_m(\mu) &:= \frac{E_{in}(\mu)}{|E|} \\
\label{gammastarDef}\gamma_m^* &:= \max_\mu \gamma_m(\mu)
\end{align}
Thus, $\gamma_m^*$ represents the maximum edges that can be captured by a partition satisfying the house quotas.  Note that $\gamma_m^*$ is highly related to other clustering metrics, such as the conductance \cite{Kannan04}, \cite{Shi00} and expansion \cite{Radicchi04}.  

We begin by noting that due to the assumption that $\sum_{h\in H} U_h(\mu) = 0$, we can separate the social welfare function into two components:
\begin{equation*}
W(\mu)  = \sum_{s \in S} U_s(\mu)  = \sum_{h \in H} \sum_{s\in \mu(h)} \left(D_h + \sum_{t\in \mu(h)} w(s,t) \right) 
 = 2 E_{in} (\mu)+\sum_{h\in H} q_h D_h.
\end{equation*}
Thus,
\begin{equation}
\frac{\max_\mu W(\mu)}{\min_\text{$\mu$ is stable} W(\mu)}
\label{PoASpec}  = \frac{Q+\max_\mu \gamma_m(\mu)}{Q+\min_{\text{stable $\mu$}} \gamma_m(\mu)}
\end{equation}
where 
\begin{equation} \label{Qdef} Q := \frac{\sum_{h\in H} q_h D_h}{2 E}. \end{equation}
Note that the parameter $Q$ is independent of the particular matching $\mu$.

Our first theorem regarding efficiency is for the ``simple'' case of unweighted social networks with equal house quotas and/or equivalently valued houses.

\begin{theorem} \label{thm:PoASimple}
Let $w(s,t) \in \left\{0,1\right\}$ for all students $s,t$ and let $q_h\geq 2, D_h \in \mathbb{Z}^+\cup\{0\}$ for all houses $h$.  If $q_h = q $ for all $h$ and/or $D_h = D$ for all $h$, then
\begin{equation*}
\min_{\text{stable }\mu } W(\mu) \geq \frac{\max_\mu W(\mu)}{1+2(m-1)\gamma_m^*}
\end{equation*}
\end{theorem}
The bound in Theorem \ref{thm:PoASimple} is tight, as illustrated by the example below.

\begin{figure} [h]  \centering
\def\svgwidth{.6\textwidth}
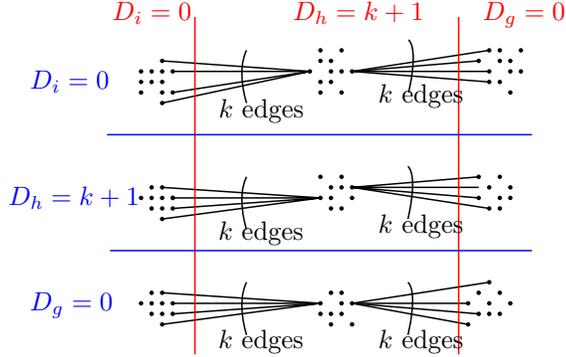 
\caption{Network that achieves PoA bound.} \label{fig:tightEx} \end{figure}

\begin{example}[Tightness of Theorem \ref{thm:PoASimple}]

Consider a setting with $m$ houses and $q_h = mk$ for all $h \in H$.  Students are grouped into clusters of size $k >2$, as shown for $m=3$ in Figure \ref{fig:tightEx}.  The houses have $D_h = k+1$ and $D_g = D_i = 0$.  Each student in the middle cluster in each row has $k$ edges to the other students outside of their cluster (but none within), as shown.

The worst-case stable exchange-matching is represented by the vertical red lines.  Note that since $D_h = k + 1$, this matching is stable, even though all edges are cut.  Thus $\min_{\mu \text{ stable}} \gamma_m(\mu) = 0$.  The optimal matching is represented by the horizontal blue lines in the figure; note that $\gamma_m^* = 1$.  To calculate the price of anarchy, we start from equations \eqref{PoASpec} and \eqref{Qdef} and calculate
\begin{equation*}
Q = \frac{\sum_{h \in H}q_h D_h}{2|E|} = \frac{mk(k+1)}{2mk(m-1)k} = \frac{k+1}{2(m-1)k},
\end{equation*}
which gives,
\begin{equation*} \frac{\max_\mu W(\mu)}{\min_{ \text{stable }\mu} W(\mu)}  = \frac{Q + \gamma_m^*}{Q + \min_{\mu \text{ stable}}\gamma_m(\mu)}  = 1 + 2(m-1)\left(\frac{k}{k+1}\right).
\end{equation*}
Notice that as $k$ becomes large, this approaches the bound of $1 + 2(m-1)\gamma_m^*$. \end{example}

We note that the requirement $q_h = q $ for all $h$ and/or $D_h = D$ for all $h$ is key to the proof of Theorem \ref{thm:PoASimple} and in obtaining such a simple bound; otherwise, Theorem \ref{thm:PoAGeneral} applies.  We omit the proofs of these theorems here for brevity; see Appendix \ref{app:PoA} for the details.

Our second theorem removes the restrictions in the theorem above, at the expense of a slightly weaker bound. Define $q_{max} = \max_{h \in H} q_h$, $w_{max} = \max_{s,t \in S} w(s,t)$ and $D_\Delta = \min_{h,g \in H} (D_h - D_g)$, assuming that the houses are ordered in increasing values of $D_h$.

\begin{theorem} \label{thm:PoAGeneral}
Let $w(s,t) \in \mathbb{R}^+\cup\{0\}$ for all students $s,t$ and $D_h \in \mathbb{R}^+\cup\{0\}$, $q_h \in \mathbb{Z}^+$ for all houses $h$, then
\begin{equation*}
\min_{ \text{stable }\mu} W(\mu) \geq \frac{\max_\mu W(\mu)}{1+2(m-1)\left(\gamma_m^*+\tfrac{q_{max}w_{max}}{D_\Delta}\right)}
\end{equation*}
\end{theorem}

Though Theorem \ref{thm:PoASimple} is tight, it is unclear at this point whether Theorem \ref{thm:PoAGeneral} is also tight.  However, a slight modification of the above example does show that it has the correct asymptotics, i.e., there exists a family of examples that have price of anarchy $\Theta(m \gamma_m^* q_{max} w_{max} D^{-1}_\Delta)$.

A first observation one can make about these theorems is that the price of anarchy has no direct dependence on the number of students.  This is an important practical observation since the number of houses is typically small, while the number of students can be quite large (similar phenomena hold in many other many-to-one matching markets).  In contrast, the theorems highlight that the degree of heterogeneity in quotas, network edge weights, and house valuations all significantly impact inefficiency.

\begin{figure*} [t]  \centering
  \subfloat[Caltech $\gamma_m^*$]{\label{fig:gammaCaltech}\includegraphics[width=0.45\textwidth]{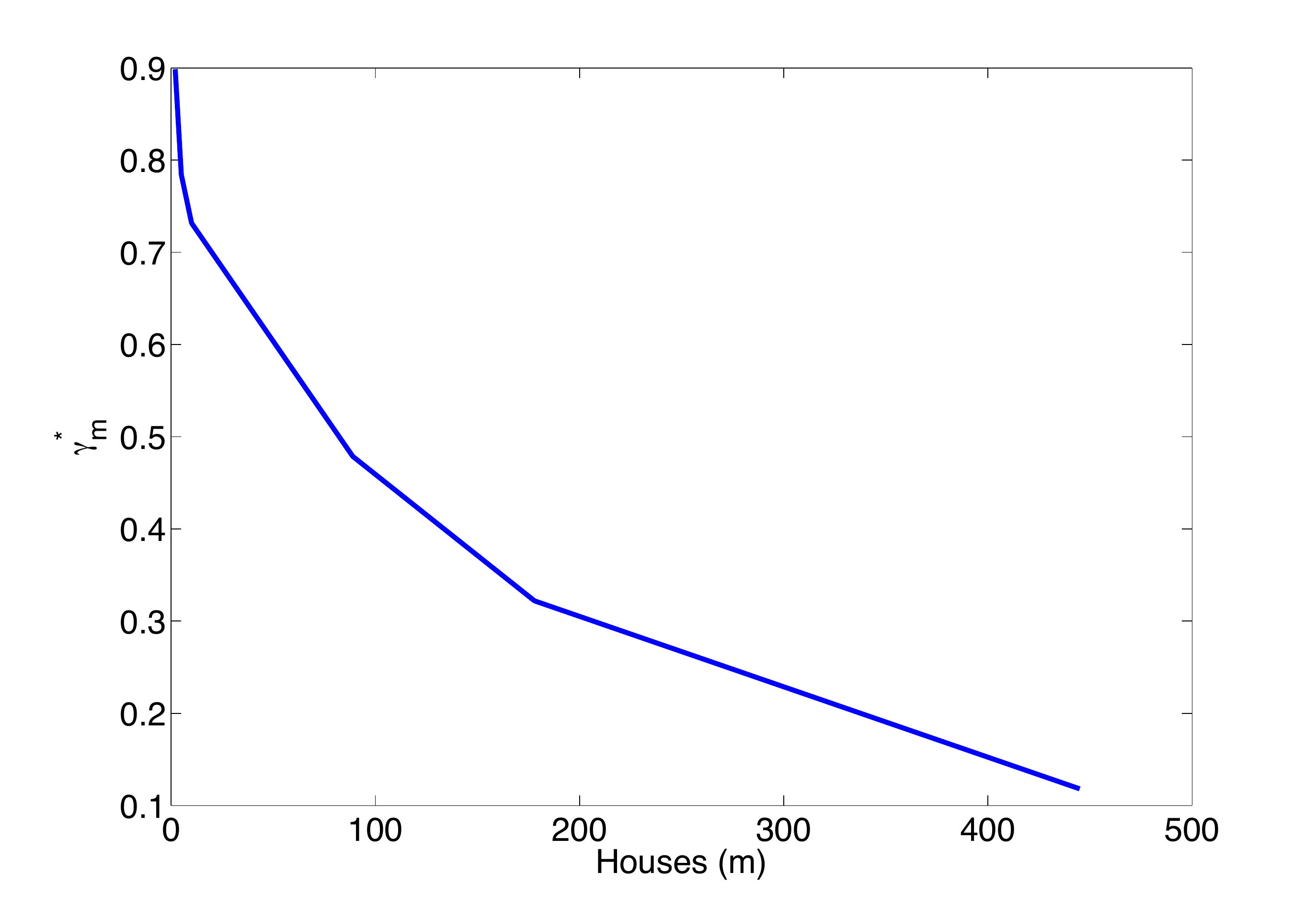}}
  \subfloat[Wikipedia $\gamma_m^*$]{\label{fig:gammaWiki}\includegraphics[width=0.45\textwidth]{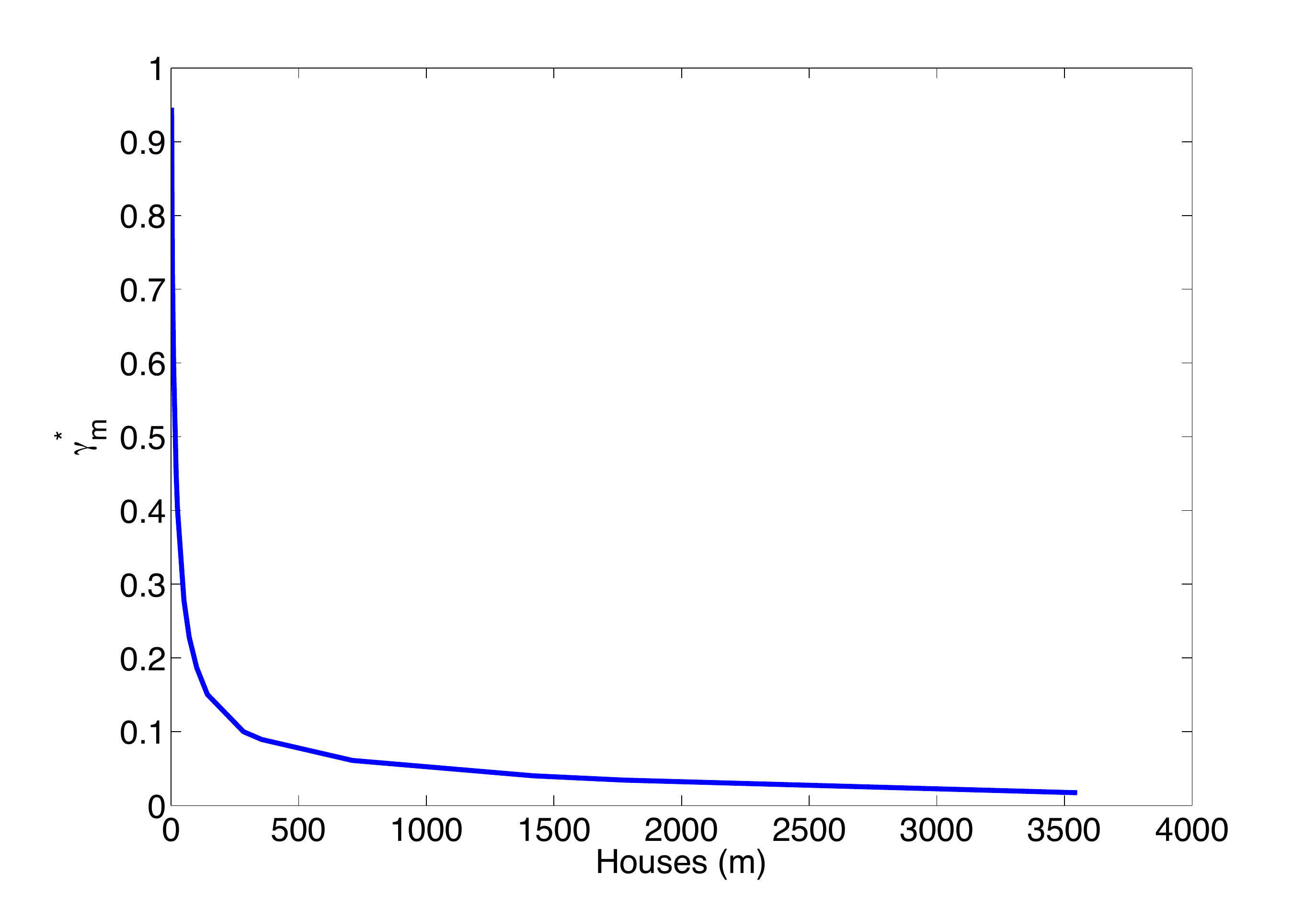}}\\
  \subfloat[Caltech PoA]{\label{fig:poaCaltech}\includegraphics[width=0.45\textwidth]{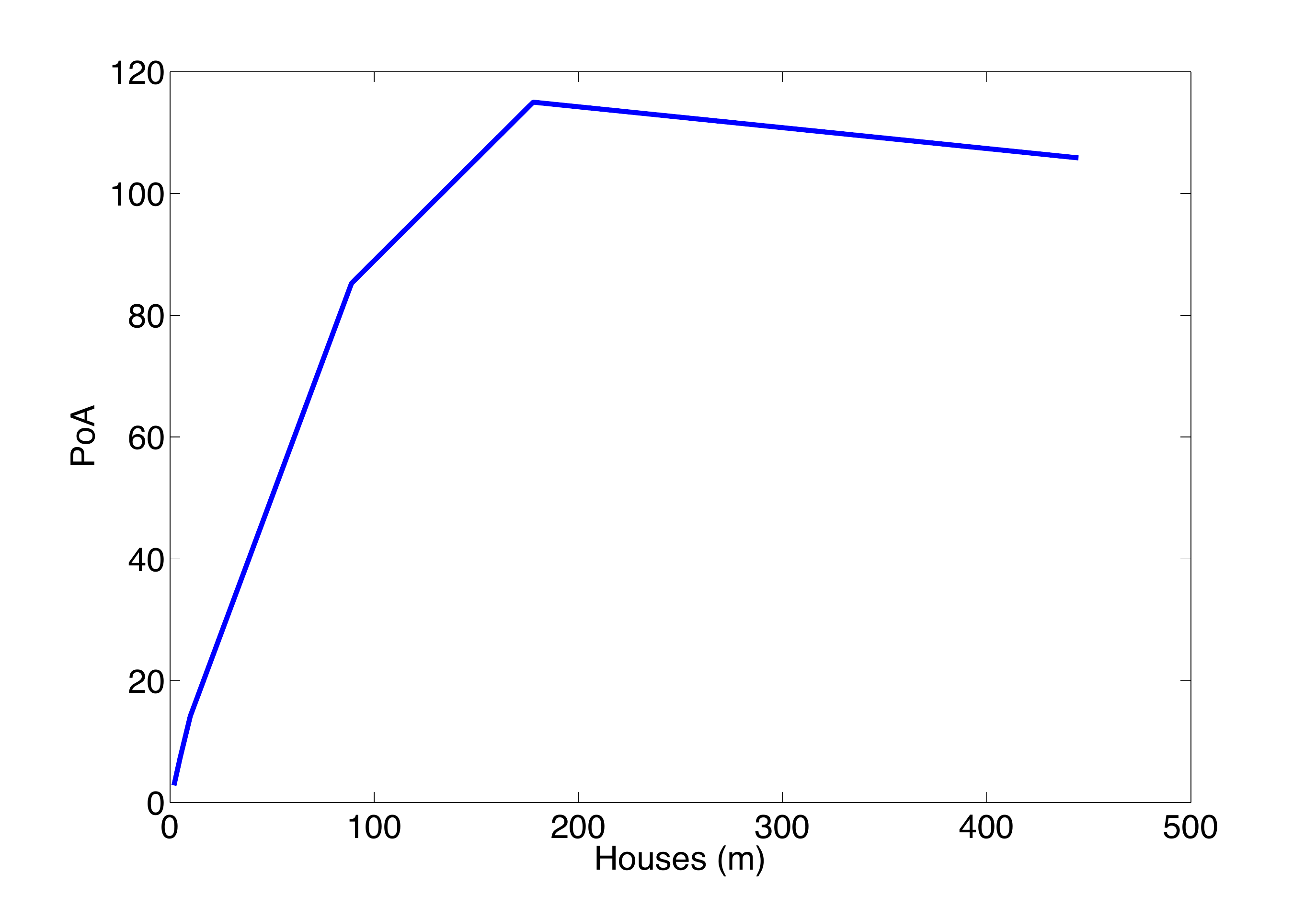}}
  \subfloat[Wikipedia PoA]{\label{fig:poaWiki}\includegraphics[width=0.45\textwidth]{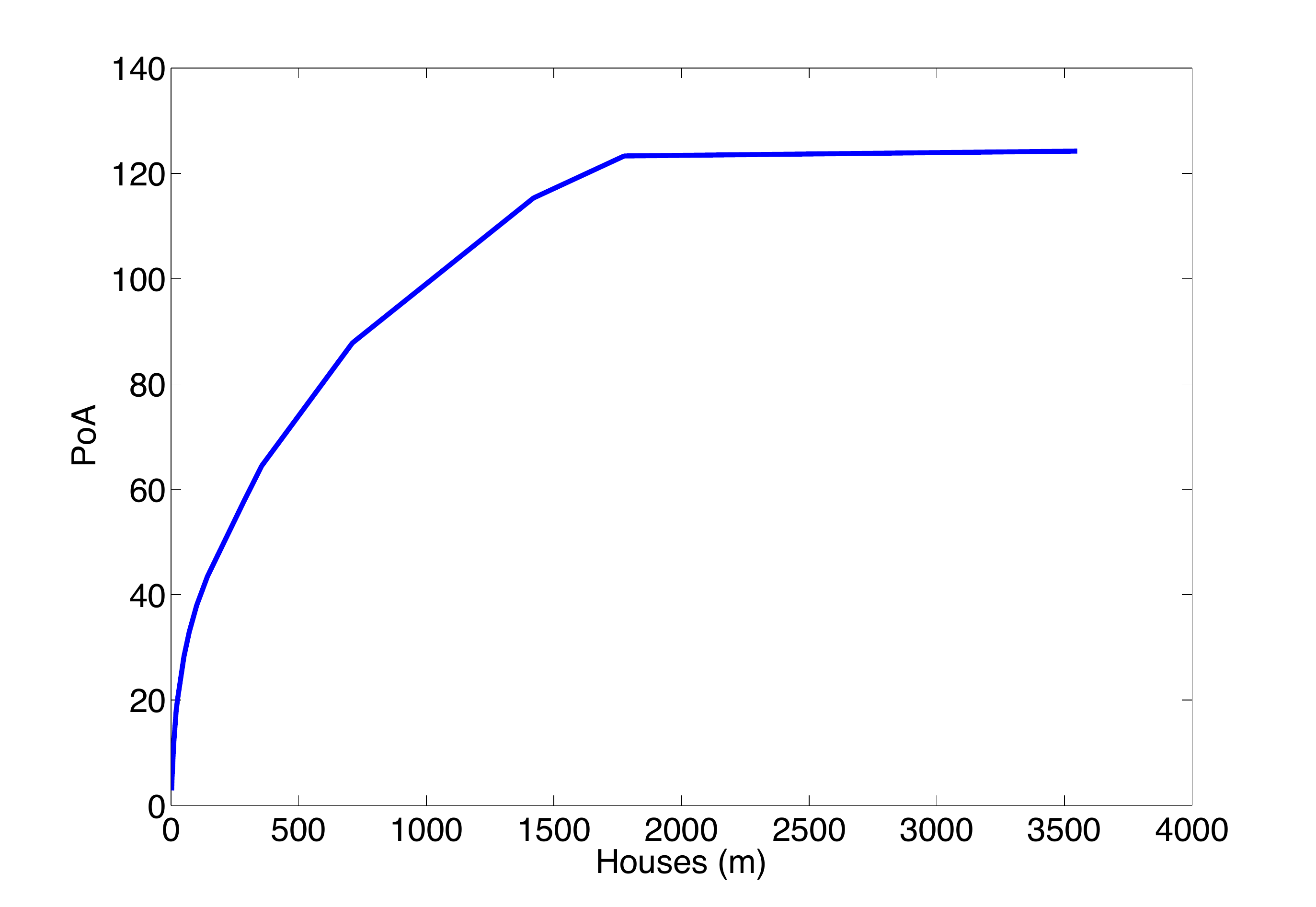}}
  \caption{Illustration of $\gamma_m^*$ and price of anarchy bounds in Theorem \ref{thm:PoASimple} for Caltech and Wikipedia networks.}
  \label{fig:gammaPoA}
\end{figure*}

A second remark about the theorems is that the only dependence on the social network is through $\gamma_m^*$, which measures how well the graph can be ``clustered'' into $m$ groups.  An important note about $\gamma_m^*$ is that it is highly dependent on $m$, and tends to shrink quickly as $m$ grows.  We give an illustration of this effect in Figures \ref{fig:gammaCaltech} and  \ref{fig:gammaWiki} using the two social network data sets described in Section \ref{sec:sim}.  A consequence of this behavior is that the price of anarchy is not actually linear in $m$ in Theorems \ref{thm:PoASimple} and \ref{thm:PoAGeneral}, as it may first appear, it turns out to be sublinear. This is illustrated in the context of real social network data in Figures \ref{fig:poaCaltech} and \ref{fig:poaWiki}.  We note that as we are increasing $m$, what we are in fact doing is creating finer allowable partitions of the network.

Next, let us consider the impact of peer effects on the price of anarchy.  Considering the simple setting of Theorem \ref{thm:PoASimple}, we see that if there were no peer effects, this would be equivalent to setting $w(s,t)=0$ for all $s,t$.  This would imply that $\gamma_m^*=0$, and so the price of anarchy is one.  Thus, another interpretation of the price of anarchy in Theorem \ref{thm:PoASimple} is the efficiency lost as a result of peer effects.
\section{Concluding remarks}

In this paper we have focused on many-to-one matchings with peer effects and complementarities.  Typically, results on this topic tend to be negative, either proving that stable matchings may not exist, e.g., \cite{Roth84,matching_book}, or that stable matchings are computationally difficult to find, e.g., \cite{Ronn90}.

In this paper, our goal has been to provide positive results. To this end, we focus on the case when peer effects are the result of an underlying social network, and this restriction on the form of the peer effects allows us to prove that a two-sided exchange-stable matching always exists and that socially optimal matchings are always stable. Further, we provide bounds on the maximal inefficiency (price of anarchy) of any exchange-stable matching and show how this inefficiency depends on the clustering properties of the social network graph.  Interestingly, in our context the price of anarchy has a dual interpretation as characterizing the degree of inefficiency caused by peer effects.

There are numerous examples of many-to-one matchings where the results in this paper can provide insight; one of particular interest to us is the matching of incoming undergraduates to residential houses which happens yearly at Caltech and other universities.  Currently incoming students only report a preference order for houses, and so are incentivized to collude with friends and not reveal their true preferences.  For such settings, the results in this paper highlight the importance of having students report not only their preference order on houses, but also a list of friends with whom they would like to be matched.  Using a combination of these factors the algorithms and efficiency bounds presented in this paper provide a promising approach, for this specific market as well as any general market where peer effects change the space of stable matchings.

The results in the current paper represent only a starting point for research into the interaction of social networks and many-to-one matchings.  There are a number of simplifying assumptions in this work which would be interesting to relax.  For example, the efficiency bounds we have proven consider only a one-sided market, where houses do not have preferences over students, students rate houses similarly, and quotas are exactly met. These assumptions are key to providing simpler bounds, and they certainly are valid in some matching markets; however relaxing these assumptions would broaden the applicability of the work greatly.

{\small
\bibliography{ref}
}
\newpage
\appendix
\section{Existence Proof} \label{app:existence}

\subsection{General case: Local maxima of $\Phi(\mu)$ are two-sided exchange-stable}
To prove Theorem \ref{thm:stable_general}, we first prove a technical lemma (Lemma \ref{lem:existencehelper_general}), from which the theorem follows immediately. The analysis is straightforward and draws its key ideas from the work of \cite{YarivWilson09}, which considers only a one-sided market rather than the two-sided market considered here.

\begin{proof}(Lemma \ref{lem:existencehelper_general})
We begin by calculating the difference in the potential function for a swap matching using \eqref{potFn}: 
\begin{align} & \Phi(\mu_s^t) - \Phi(\mu) = \sum_{h \in H} U_h(\mu_s^t) - U_h(\mu) + \sum_{s \in S} D^s_{\mu_s^t(s)} - D^s_{\mu(s)} \notag \\
& \; + \frac{1}{2} \sum_{s \in S} \left( \sum_{x \in \mu_s^t(s)} w(s,x) \right) - \frac{1}{2} \sum_{s \in S} \left( \sum_{x \in \mu(s)} w(s,x) \right) \end{align}
Expanding and canceling like terms, we have
\begin{align} & \Phi(\mu_s^t) - \Phi(\mu) =   U_h(\mu_s^t) - U_h(\mu) + U_g(\mu_s^t) - U_g(\mu) 
+ D_g^s - D_g^t + D_h^t - D_h^s \notag \\
& \; + \frac{1}{2} \left[ \left( \sum_{x \in g} w(s,x) - w(s,t) + \sum_{x \in h} w(t,x) - w(s,t) \right) \right. \notag \\
& \;\;\;\;\; \left. + \left( \sum_{x \in g} w(x,s)-w(s,t) + \sum_{x \in h} w(x,t) - w(s,t)  \right) \right]  \notag \\
& \; - \frac{1}{2} \left[ \left(\sum_{x \in h} w(s,x) + \sum_{x \in g} w(t,x) \right) + \left( \sum_{x \in h} w(x,s) + \sum_{x \in g} w(x,t) \right) \right] \end{align}
which becomes, due to the symmetry of the social network,
\begin{align} \Phi(\mu_s^t) - \Phi(\mu)&  =  U_h(\mu_s^t) - U_h(\mu) + U_g(\mu_s^t) - U_g(\mu) 
+ D_g^s - D_g^t + D_h^t - D_h^s \notag \\
& + \sum_{x \in g} \left( w(s,x) - w(t,x) \right) + \sum_{x \in h} \left( w(t,x)-w(s,x) \right) - 2 w(s,t) \label{phiDif} \end{align}
Note that if $t$ is a ``hole'', this becomes
\begin{align} \Phi(\mu_s^t) - \Phi(\mu) &  =  U_h(\mu_s^t) - U_h(\mu) + U_g(\mu_s^t) - U_g(\mu) 
+ D_g^s - D_h^s \notag \\ & + \sum_{x \in g} \left( w(s,x) \right) - \sum_{x \in h} \left( w(s,x) \right)  \end{align}

Now, consider a matching $\mu$ and a swap matching $\mu_s^t$ that satisfies (i) and (ii) from the lemma statement.  Without loss of generality, assume that student $s$ strictly improves.  The other student could be either a ``hole'' or a real student that either improves or is indifferent to the swap.  The other cases are symmetric.  Define $\mu(s) = h,$ and $\mu(t) = g$. The change in utility for student $s$ is then
\begin{align*}
0 &< U_s(\mu_s^t) - U_s(\mu) \\
&= D^s_g - D^s_h - \sum_{x \in \mu(h)} w(s,x) + \sum_{x \in \mu(g)} w(s,x) - w(s,t),
\end{align*}
Similarly, for student $t$, we have
\begin{align*}
0 &\leq U_t(\mu_s^t) - U_t(\mu) \\
& = D^t_h - D^t_g - \sum_{x \in \mu(g)} w(t,x) + \sum_{x \in \mu(h)} w(t,x) - w(s,t).
\end{align*}
Adding the above inequalities, we obtain the following:
\begin{align*}
0 &<  D_g^s - D_g^t + D_h^t - D_h^s +  \sum_{x \in \mu(g)} \left(w(s,x) - w(t,x)\right)  + \sum_{x \in \mu(h)} \left(w(t,x) - w(s,x)\right) \\
&\quad -2w(s,t) \\
&:=\: \delta_{s,t}
\end{align*}

On the house side of the market, we have
\begin{equation*}
0 \leq  U_h(\mu_s^t) - U_h(\mu) + U_g(\mu_s^t) - U_g(\mu) := \Delta_H
\end{equation*}
as only houses $h$ and $g$ are affected by the swap and the change in their utilities is non-negative by assumption.  

Thus, using \eqref{phiDif} above, we have 
\begin{equation} \Phi(\mu_s^t) - \Phi(\mu) = \Delta_H + \delta_{s,t} > 0. \end{equation} 
Note that this holds even if $t$ is a ``hole.''   \qed \end{proof}

\begin{proof}(Theorem \ref{thm:stable_general})
Let matching $\mu$ be a local maximum of $\Phi(\mu)$.  Assume, by way of contradiction, that $\mu$ is not two-sided exchange-stable.  Lemma \ref{lem:existencehelper_general} shows that any swap matching that is acceptable to all parties (i.e. satisfies conditions (i) and (ii)) strictly increases $\Phi(\mu)$.  But this contradicts the assumption that $\mu$ is a local maximum. Thus, $\mu$ must be two-sided exchange-stable. \qed
\end{proof}

As the number of matches is finite, the global maximum of the potential function must be two-sided exchange-stable, and, therefore, a two-sided exchange stable matching will always exist.  

\subsection{Specific case: Local maxima of $W(\mu)$ are two-sided exchange-stable}
In this specific case, we assume that house quotas are exactly met (i.e., there are no vacancies or ``holes'') and that all students rate houses according to the same rules (i.e., $D_h^s = D_h^t \; \forall \; s \neq t$). 

\begin{proof}(Lemma \ref{lem:existencehelper})
Consider a matching $\mu$ and a swap matching $\mu_s^t$ that satisfies (i) and (ii) from the lemma statement.  Note that due to the assumption that the house quotas are all met, the swap must be between two students, not a student and a ``hole.''  Without loss of generality, assume that student $s$ strictly improves. The other cases are symmetric.  Define $\mu(s) = h,$ and $\mu(t) = g$. The change in utility for student $s$ is then
\begin{align*}
0 &< U_s(\mu_s^t) - U_s(\mu) \\
&= D_g - D_h - \sum_{x \in \mu(h)} w(s,x) + \sum_{x \in \mu(g)} w(s,x) - w(s,t),
\end{align*}
Similarly, for student $t$, we have
\begin{align*}
0 &\leq U_t(\mu_s^t) - U_t(\mu) \\
& = D_h - D_g - \sum_{x \in \mu(g)} w(t,x) + \sum_{x \in \mu(h)} w(t,x) - w(s,t).
\end{align*}
Adding the above inequalities, we obtain the following:
\begin{align*}
0 &<  \sum_{x \in \mu(g)} \left(w(s,x) - w(t,x)\right)  + \sum_{x \in \mu(h)} \left(w(t,x) - w(s,x)\right) \\
&\quad -2w(s,t) \\
&:=\: \delta_{s,t}
\end{align*}
Continuing, the total change in utility for \emph{all} students is:
\begin{align}
\Delta_S :=& \sum_{x\in S} U_x(\mu_s^t)-\sum_{x\in S} U_x(\mu) \notag  \\
=& \: \delta_{s,t}  + \underbrace{\sum_{x \in \mu(g)} w(x,s) - w(s,t)}_\textrm{gain from $s$ joining $g$} - \underbrace{\sum_{x \in \mu(h)} w(x,s)}_\textrm{loss from $s$ leaving $h$} \notag  \\
& \quad + \underbrace{\sum_{x \in \mu(h)} w(x,t) - w(s,t)}_\textrm{gain from $t$ joining $h$} - \underbrace{\sum_{x \in \mu(g)} w(x,t)}_\textrm{loss from $t$ leaving $g$} \notag \\
=& \: 2 \delta_{s,t}  \label{symm} \\
>&  \: 0 \notag
\end{align}
where line (\ref{symm}) comes from the fact that we assume the social network graph is symmetric.

On the house side of the market, we have
\begin{equation*}
0 \leq  U_h(\mu_s^t) - U_h(\mu) + U_g(\mu_s^t) - U_g(\mu) := \Delta_H
\end{equation*}
as only houses $h$ and $g$ are affected by the swap and the change in their utilities is non-negative by assumption.  Thus, the total social welfare strictly increases:
\begin{equation*}
W(\mu_s^t) - W(\mu) = \Delta_S + \Delta_H > 0
\end{equation*}
\qed \end{proof}

Using Lemma \ref{lem:existencehelper} it is now easy to complete the proof of Theorem \ref{thm:stable_SW}.

\begin{proof}(Theorem \ref{thm:stable_SW})
Let matching $\mu$ be a local maximum of $W(\mu)$.  Lemma \ref{lem:existencehelper} shows that any swap matching that is acceptable to all parties (i.e. satisfies conditions (i) and (ii)) strictly increases the total social welfare.  But this contradicts the assumption that $\mu$ is a local maximum. Thus, $\mu$ must be stable.
\qed \end{proof}

\section{Proofs of PoA Theorems} \label{app:PoA}
We note that these proofs hold for the one-sided market; i.e., when $U_h(\mu) = 0\;\forall h \in H$, where the quotas for the houses are \emph{exactly} satisfied; i.e., there are no ``holes,'' and students value houses according to the same rules; i.e.,  $D_h^s = D_h^t \;  \forall \; s \neq t, \; h \in H$.  Also note that for ease of notation, we use $E$ instead of $|E|$ to represent the total edge weight of the graph in this appendix.

\subsection{Proof of Theorem \ref{thm:PoASimple}} \label{subsec:unweighted}
Throughout the proof we assume that the houses are ordered: i.e., if $g < h$ then $D_g < D_h$. An important tool that we use throughout the proof is a rephrasing of the definition of exchange stability \emph{in the one-sided market case} in terms of a function $\alpha$ as follows.

\begin{definition}
Let $\alpha_\mu(s,g)$ be a function representing the benefit a student $s$ gains by moving to house $g$ under matching $\mu$:
\begin{equation} \label{alpha} \alpha_\mu(s,g) = D_g - D_{\mu(s)} + \sum_{x \in \mu(g)}w(s,x) - \sum_{x \in \mu^2(s)}w(s,x) \end{equation}
\end{definition}

Notice that using the definition above, given a specific swap matching $\mu_s^t$ where $t \in \mu(g)$, we can calculate the difference in utility for the involved student $s$ as
\begin{equation*} U_s(\mu_s^t) - U_s(\mu) = \alpha_\mu(s,g) - w(s,t) \end{equation*}
because $\sum_{x \in \mu_s^t(g)}w(s,x) = \sum_{x \in \mu(g)}w(s,x) - w(s,t)$.

The definition of $\alpha$ also provides a useful new phrasing of the definition of exchange stability, which is equivalent to that of Definition \ref{psDef_orig} when the market is one-sided, i.e, when $U_h(\mu) = 0 \; \forall h \in H$.  Note that we are only considering the Price of Anarchy for the one-sided market here -- we plan to generalize these results for the two-sided case in future work.

\begin{definition} \label{psDef}
A matching $\mu$ is \textbf{exchange stable (ES)} in the one-sided (students-only) housing assignment market if and only if for all pairs of students $s \in \mu(h)$ and $t \in \mu(g)$, at least one of the following conditions holds:
\begin{description}
  \item[Condition 1:] $s$ doesn't want to swap, i.e.,  $\alpha_\mu(s,g) < w(s,t).$
  \item[Condition 2:] $t$ doesn't want to swap, i.e., $\alpha_\mu(t,h) < w(s,t).$
  \item[Condition 3:] $s$ and $t$ are indifferent, i.e., $\alpha_\mu(s,g) =$\\ $ \alpha_\mu(t,h) = w(s,t).$
\end{description}
\end{definition}

Using the above rephrasing of the definition of exchange stability, we now continue with the proof of Theorem \ref{thm:PoASimple}.  In order to prove an upper bound on the price of anarchy, we prove a lower bound on $\gamma_m(\mu)$ when $\mu$ is stable.  To prove this lower bound, we first prove an upper bound on the number of cross edges ($E_{hg}=E_{gh}$) in the following lemma.

\begin{lemma} \label{E_hg_bound}
Let $w(s,t) \in \left\{0,1\right\}$ for all students $s,t$ and let $q_h\geq 2, D_h \in \mathbb{Z}^+\cup\{0\}$ for all $h$.  Let $q_h = q$ for all $h$ and/or $D_h = D$ for all $h$. If a matching $\mu$ is stable, then for all houses $h$ and $g$,
\begin{equation}
E_{hg} \leq \max(q_h (D_h - D_g),q_g(D_g - D_h)) + 2(E_{hh} + E_{gg})
\end{equation}
\end{lemma}

\begin{proof}
Using the conditions of stability from Definition \ref{psDef} and Lemmas \ref{E_hg_1} and \ref{E_hg_3} as summarized below and proved in Appendix \ref{appProofs}, we have
\\
\emph{Case 1:}
If there exists $s \in \mu(h)$ such that $\alpha_\mu(s,g) > 1$ then, by Lemma \ref{E_hg_1}, if $\mu$ is stable it follows that
\begin{equation*} E_{gh} \leq q_g (D_g-D_h) + 2 E_{gg} \end{equation*}
\emph{Case 2:}
If there exists $t \in \mu(g)$ such that $\alpha_\mu(t,h) > 1$ then, by Lemma \ref{E_hg_1}, if $\mu$ is stable it follows that
\begin{equation*} E_{hg} \leq q_h (D_h-D_g) + 2 E_{hh} \end{equation*}
\emph{Case 3:}
If there does not exist $s \in \mu(h)$ such that $\alpha_\mu(s,g) > 1$ and there does not exist $t \in \mu(g)$ such that $\alpha_\mu(t,h) > 1$ then, by Lemma \ref{E_hg_3}, if $\mu$ is stable it follows that
\begin{align*} E_{hg} \leq & \max (q_h(D_h-D_g),q_g(D_g-D_h)) + 2 (E_{hh} + E_{gg}) \end{align*}

Given any matching $\mu$ in the student-only market, it must fall into one of the three cases above.  Thus, if $\mu$ is stable, it follows that one of the three bounds above holds.  Because the edges are undirected, $E_{hg} = E_{gh}$, we can combine the three bounds to conclude that if $\mu$ is stable,
\begin{equation*}
E_{hg}  \leq  \max(q_h(D_h-D_g),q_g(D_g-D_h)) + 2 (E_{hh} + E_{gg})
\end{equation*}
\end{proof}

\noindent Next, we use the above to prove a lower bound on $\gamma_m(\mu)$.

\begin{lemma}  \label{gammaBound}
Let $w(s,t) \in \left\{0,1\right\}$ and let $q_h\geq 2, D_h \in \mathbb{Z}^+\cup\{0\}$ for all $h$.  Let $q_h = q$ for all $h$ and/or $D_h = D$ for all $h$. If a matching $\mu$ is stable, then
\begin{equation}
\gamma_m(\mu) \geq \max \left( \frac{E-\sum_{g<h} q_h (D_h-D_g)}{(2m-1)E}, 0 \right)
\end{equation}
\end{lemma}

\begin{proof}
\begin{align}
E_{in} (\mu) & = E - \sum_{g<h} E_{gh} \notag \\
& \geq E - \sum_{g<h} \left( q_h(D_h-D_g)+ 2 (E_{hh} + E_{gg}) \right) \label{maxBound} \\
& = E - 2(m-1)E_{in} (\mu) - \sum_{g<h} \left( q_h(D_h-D_g) \right) \notag
\end{align}
where we have used the assumption that the houses are ordered in line \eqref{maxBound}.
Solving for $E_{in} (\mu)$ gives
\begin{equation*}
E_{in} (\mu) \geq \frac{E-\sum_{g<h}q_h(D_h-D_g)}{2m-1}.
\end{equation*}
Thus,
\begin{equation*}
\gamma_m(\mu) = \frac{E_{in} (\mu)}{E} \geq \frac{E-\sum_{g<h}q_h(D_h-D_g)}{(2m-1)E}.
\end{equation*}
Note that the above bound is only useful when the numerator is positive; otherwise, the bound becomes negative.  However, it is immediate to see that $\gamma_m(\mu) \geq 0$ always, as $E_{in} (\mu)$ and $E$ are non-negative, which completes the proof.
\end{proof}

Finally, we can complete the proof of Theorem \ref{thm:PoASimple} using the above lemmas.  There are two cases to consider, depending on the value of $E:\\$

\emph{Case 1:} $E > \sum_{g>h} q_h(D_h-D_g) $\\

Plugging the bound from Lemma \ref{gammaBound} into \eqref{PoASpec} gives
\begin{align*}
\frac{\max_\mu W(\mu)}{\min_\text{$\mu$ is stable} W(\mu)} & = \frac{Q + \gamma_m^*}{Q + \gamma_m(\mu)} \\
& \leq \tfrac{\tfrac{\sum_{h \in H}q_h D_h}{2E} + \gamma_m^*}{\tfrac{\sum_{h \in H}q_h D_h}{2E} + \tfrac{E-\sum_{g<h}q_h (D_h-D_g)}{(2m-1)E}} \\
& = \tfrac{(2m-1)\sum_{h \in H}{q_h D_h} + 2(2m-1)E\gamma_m^*}{(2m-1)\sum_{h \in H}{q_h D_h} + 2E-2\sum_{g<h}{q_h (D_h-D_g)}}
\end{align*}
Using Lemma \ref{m-1 bound} from the Appendix to substitute for $\sum_{h \in H}{q_h D_h}$ is then enough to complete the proof in this case, after some algebra using the fact that $\gamma_m^* \leq 1$.

\emph{Case 2:} $E \leq \sum_{g<h} q_h(D_h-D_g)\\$

In this case, Lemma \ref{gammaBound} states that $\gamma_m(\mu) \geq 0$.  Using this bound and plugging into \eqref{PoASpec} gives
\begin{equation} \label{eq:PoAQ}
\frac{\max_\mu W(\mu)}{\min_\text{$\mu$ is stable} W(\mu)} = \frac{Q + \gamma_m^*}{Q + \gamma_m(\mu)} \leq  1 + \frac{\gamma_m^*}{Q}
\end{equation}

Note that $Q > 0$ as long as $E > 0$ because we are given that $E \leq \sum_{g<h} q_h(D_h-D_g)$ in this case. Further, note that the case of $E = 0$ is trivial because all matchings have the same welfare and so the price of anarchy is $1$.

Using $E\leq \sum_{g<h} q_h(D_h-D_g)$ we have
\begin{equation} \label{eq:QBound}
Q \geq \frac{\sum_{h \in H}q_h D_h}{2 \sum_{g<h} q_h(D_h-D_g)}.
\end{equation}
Combining \eqref{eq:PoAQ} and \eqref{eq:QBound} and again using Lemma \ref{m-1 bound} from the appendix is then enough to complete the proof in this case, after some algebra.

One final remark about this proof is that in the special case of $D_h = 0$ a tighter bound holds. Specifically, the price of anarchy is bounded by $(2m-1) \gamma_m^*$ in this case.

\subsection{Proof of Theorem \ref{thm:PoAGeneral}} \label{subsec:weighted}
The proof of Theorem \ref{thm:PoAGeneral} follows the same structure as the proof of Theorem \ref{thm:PoASimple}, with a few added complexities that cause the bound to become weaker.

To begin, we again derive a bound on the cross-edges.

\begin{lemma} \label{E_hg_bound_general}
Let $w(s,t) \in \mathbb{R}^+\cup\{0\}$ for all students $s,t$ and let $D_h \in \mathbb{R}^+\cup\{0\}$ for all houses $h$.   If a matching $\mu$ is stable,  then for all houses $h$ and $g$,
\begin{align*}
E_{hg} \leq & \max(q_h (D_h - D_g),q_g(D_g - D_h)) + 2(E_{hh} + E_{gg}) \\
& + q_{max} w_{max}.
\end{align*}
\end{lemma}

\begin{proof}
Using the conditions of stability from Definition \ref{psDef} for the one-sided market and Lemmas \ref{E_hg_1_general} and \ref{E_hg_3_general} from Appendix \ref{appProofs}, we have three cases.
\\
\emph{Case 1:}
If there exists $s \in \mu(h)$ such that $\alpha_\mu(s,g) > w(s,t)$ for all $t \in \mu(g)$ then, by Lemma \ref{E_hg_1_general}, if $\mu$ is stable, it follows that
\begin{equation*} E_{gh} \leq q_g (D_g-D_h) + 2 E_{gg} + q_g w_{max}. \end{equation*}
\emph{Case 2:}
If there exists $t \in \mu(g)$ such that $\alpha_\mu(t,h) > w(s,t)$ for all $s \in \mu(h)$ then, by Lemma \ref{E_hg_1_general}, if $\mu$ is stable, it follows that
\begin{equation*} E_{hg} \leq q_h (D_h-D_g) + 2 E_{hh} +q_h w_{max}.\end{equation*}
\emph{Case 3:}
If there does not exist $s \in \mu(h)$ such that $\alpha_\mu(s,g) > w(s,t)$ and there does not exist $t \in \mu(g)$ such that $\alpha_\mu(t,h) > w(s,t)$, for all $t \in \mu(g), s \in \mu(h)$ respectively, then,  by Lemma \ref{E_hg_3_general},  if $\mu$ is stable, it follows that
\begin{align*} E_{hg} \leq  & \max (q_h(D_h-D_g),q_g(D_g-D_h))  + 2 (E_{hh} + E_{gg}) \\
& + q_{max} w_{max}. \end{align*}

Given any matching $\mu$, it must fall into one of the three cases above.  Thus, if $\mu$ is exchange-stable, it follows that one of the three bounds above holds.  Because the edges are undirected, $E_{hg} = E_{gh}$, we can combine the three bounds to conclude that, if $\mu$ is stable,
\begin{align*}
E_{hg}  \leq & \max(q_h(D_h-D_g),q_g(D_g-D_h)) + 2 (E_{hh} + E_{gg}) \\
&  + q_{max} w_{max}
\end{align*} \end{proof}

\noindent Next, we use the above to prove a lower bound on $\gamma_m(\mu)$.

\begin{lemma}  \label{gammaBound_general}
Let $w(s,t) \in \mathbb{R}^+\cup\{0\}$.  If a matching $\mu$ is stable, then
\begin{equation*}
\gamma_m(\mu) \geq \max \left( \tfrac{E-\sum_{g<h} q_h (D_h-D_g) - {m \choose 2} q_{max} w_{max}}{(2m-1)E}, 0 \right)
\end{equation*}
\end{lemma}

\begin{proof}
\begin{align}
&E_{in} (\mu) = E - \sum_{g<h} E_{gh} \notag \\
& \quad \geq E - \sum_{g<h}( q_h(D_h-D_g)+ 2 (E_{hh} + E_{gg})  + q_{max} w_{max} ) \label{maxBound_gen} \\
& \quad = E - 2(m-1)E_{in} (\mu) - \sum_{g<h} \left( q_h(D_h-D_g) \right)\notag \\
& \qquad- {m \choose 2} q_{max} w_{max}\notag
\end{align}
where line \eqref{maxBound_gen} follows from the assumption that the houses are ordered.

Solving for $E_{in} (\mu)$ gives
\begin{equation*}
E_{in} (\mu) \geq \tfrac{E-\sum_{g<h}q_h(D_h-D_g) - {m \choose 2} q_{max} w_{max} }{2m-1},  \end{equation*}
and thus
\begin{equation*}
\gamma_m(\mu) = \frac{E_{in} (\mu)}{E} \geq \tfrac{E-\sum_{g<h}q_h(D_h-D_g) - {m \choose 2} q_{max} w_{max} }{(2m-1)E}.
\end{equation*}

This bound is only relevant when $E > \sum_{g>h} q_h(D_h-D_g)+ {m \choose 2} q_{max} w_{max}$. Otherwise, the bound becomes negative, in which case we use the fact that  $\gamma_m(\mu) \geq 0$ always, as $E_{in} (\mu)$ and $E$ are non-negative.
\end{proof}

Finally, we can complete the proof of Theorem \ref{thm:PoAGeneral} using the lemmas above.  There are two cases to consider, depending on the value of $E$.

\emph{Case 1: $E > \sum_{g<h} q_h (D_h - D_g) + {m \choose 2} q_{max} w_{max}$}

Plugging the bound from Lemma \ref{gammaBound_general} into \eqref{PoASpec} gives
\begin{align*}
&\tfrac{\max_\mu W(\mu)}{\min_\text{$\mu$ is stable} W(\mu)} = \frac{Q + \gamma_m^*}{Q + \gamma_m(\mu)} \\
& \quad \leq \tfrac{\tfrac{\sum_{h \in H}q_h D_h}{2E} + \gamma_m^*}{\tfrac{\sum_{h \in H}q_h D_h}{2E} + \tfrac{E-\sum_{g<h}q_h (D_h-D_g) - {m \choose 2} q_{max} w_{max}}{(2m-1)E}} \\
& \quad = \tfrac{(2m-1)\sum_{h \in H}{q_h D_h} + 2(2m-1)E\gamma_m^*}{(2m-1)\sum_{h \in H}{q_h D_h} + 2E-2\sum_{g<h}{q_h (D_h-D_g)} - 2{m \choose 2} q_{max} w_{max}}
\end{align*}

Using Lemma \ref{m-1 bound} to substitute for $\sum_{h \in H}{q_h D_h}$, the bound becomes, after some algebra,
\begin{align*}
& \tfrac{\max_\mu W(\mu)}{\min_\text{$\mu$ is stable} W(\mu)} \\
& \quad \leq (1+2(m-1)\gamma_m^*) \times \\
& \qquad \left( \tfrac{2(m-1)E + \sum_{g<h} q_h(D_h-D_g)}{2(m-1) E + \sum_{g<h}q_h(D_h-D_g) - 2(m-1) {m \choose 2} q_{max} w_{max}}\right).
\end{align*}

Using $E \geq \sum_{g<h} q_h (D_h - D_g) + {m \choose 2} q_{max} w_{max}$, we have, after some algebra,
\begin{align*}
& \tfrac{\max_\mu W(\mu)}{\min_\text{$\mu$ is stable} W(\mu)} \\
& \quad \leq (1+2(m-1)\gamma_m^*)  \left( 1+ \tfrac{2(m-1) {m \choose 2} q_{max} w_{max}}{(2m-1) \sum_{g<h} q_h(D_h - D_g)}\right).
\end{align*}

\emph{Case 2: $E \leq \sum_{g<h} q_h (D_h - D_g) + {m \choose 2} q_{max} w_{max}$}

In this case, Lemma \ref{gammaBound_general} states that $\gamma_m(\mu) \geq 0$. Using this bound and plugging into \eqref{PoASpec}, we have
\begin{equation*}
\tfrac{\max_\mu W(\mu)}{\min_\text{$\mu$ is stable} W(\mu)}  = \frac{Q + \gamma_m^*}{Q + \gamma_m(\mu)} \leq  1 + \frac{\gamma_m^*}{Q}.
\end{equation*}

Using $E \leq \sum_{g<h} q(D_h-D_g) + {m \choose 2} q_{max} w_{max}$ we have
\begin{equation*}
Q \geq \frac{\sum_{h \in H}q_h D_h}{2 \sum_{g<h} q_h(D_h-D_g) + 2{m \choose 2} q_{max} w_{max}}.
\end{equation*}
and so the price of anarchy becomes, again using Lemma \ref{m-1 bound},
\begin{equation*}
\tfrac{\max_\mu W(\mu)}{\min_\text{$\mu$ is stable} W(\mu)}  
\leq 1 + 2(m-1)\gamma_m^* + \frac{2{m \choose 2}   q_{max}w_{max}}{\sum_{h \in H} q_h D_h}.
\end{equation*}
We can combine the two cases into one (looser) bound,
\begin{equation*}
\tfrac{\max_\mu W(\mu)}{\min_\text{$\mu$ is stable} W(\mu)}  \leq 1+ 2(m-1)\gamma_m^* + \frac{2 (m-1) q_{max}w_{max}}{D_\Delta}.
\end{equation*} 

\section{Technical Lemmas}
\label{appProofs}

This appendix includes the lemmas used in the proofs of Theorems \ref{thm:PoASimple} and \ref{thm:PoAGeneral}.

\begin{lemma} \label{E_hg_1}
Let $w(s,t) \in \left\{0,1\right\}$ for all students $s,t$ and let $D_h \in \mathbb{Z}^+\cup\{0\}$ for all $h$. Let  $\mu$ be a stable matching.  If there exists a student $s \in \mu(h)$ such that $\alpha_\mu(s,g) > 1$ for some other house $g$, then
$E_{gh} \leq q_g(D_g-D_h) + 2 E_{gg}$.
\end{lemma}

\begin{proof}
Since $\mu$ is stable, then for all $t\in \mu(g)$, $(s,t)$ must satisfy at least one of the three conditions stated in the definition of exchange stability (Definition \ref{psDef}). However, for all $t\in\mu(g)$,
\begin{equation*}
\alpha_\mu(s,g) > 1 \geq w(s,t).
\end{equation*}
Thus, $(s,t)$ cannot satisfy conditions 1 or 3. Therefore, it must satisfy condition 2, which implies that for all $t\in\mu(g)$
\begin{equation*}
\alpha_\mu(t,h) < w(s,t) \leq 1.
\end{equation*}
Since $D_h,w(s,t) \in \mathbb{Z}^+\cup\{0\}$ we have that $\alpha_\mu(t,h)\in\mathbb{Z}^+\cup\{0\}$, and so
\begin{equation*} \alpha_\mu(t,h) < 1 \implies \alpha_\mu(t,h) \leq 0, \;\forall\; t \in \mu(g).
\end{equation*}
Summing over all $t \in \mu(g)$ gives
\begin{equation*} \sum_{t \in \mu(g)}\alpha_\mu(t,h) \leq 0. \end{equation*}
Using the definition of $\alpha$, we have
\begin{equation*}
\sum_{t \in \mu(g)}{\left(D_h - D_g + \sum_{x \in \mu(h)}w(t,x) - \sum_{x \in \mu(g)}w(t,x)\right)} \leq 0.  \end{equation*}
Simplifying the above yields
\begin{equation*}
q_g(D_h - D_g) + E_{gh} - 2E_{gg} \leq 0,
\end{equation*}
from which the desired bound follows. \end{proof}

%
%

\begin{lemma} \label{E_hg_3}
Let $w(s,t) \in \left\{0,1\right\}$ for all students $s,t$, and let $D_h \in \mathbb{Z}^+\cup\{0\}$ for all houses $h$.  Let $\mu$ be a stable matching and let $q_h = q \geq 2$ and/or $D_h = D \in \mathbb{Z}^+\cup\{0\}$ for all $h$. If (i) there does not exist an $s \in \mu(h)$ such that $\alpha_\mu(s,g) > 1$ and (ii) there does not exists a $\in \mu(g)$ such that $\alpha_\mu(t,h) > 1$,  then
\begin{equation*} E_{hg} \leq  \max(q_h(D_h-D_g), q_g(D_g-D_h))
 + 2 (E_{hh} + E_{gg})
 \end{equation*}
\end{lemma}


\begin{proof}
It follows from the assumptions in the theorem statement that the students in houses $h$ and $g$ can be partitioned into 6 sets based on their house and $\alpha$ values (either 1, 0, or negative), as shown in Figure \ref{fig:alphaPart}.

\begin{figure} [ht!] \centering
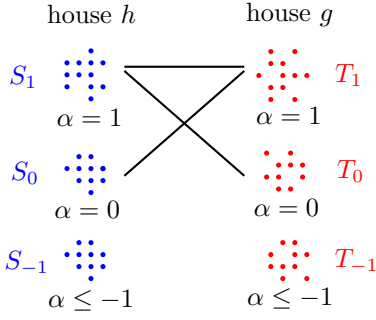
\caption{Partition of students based on $\alpha$ function} \label{fig:alphaPart} \end{figure}

Let $S_0$, $S_1$, and $S_{-1}$ denote the set of students in house $h$ such that $\alpha_\mu(s,g) = 0$, $\alpha_\mu(s,g) = 1$, and $\alpha_\mu(s,g) \leq -1$ respectively. For convenience, we use the same notation for the set and the number of students in the set, e.g., $|S_1| = S_1$  The same conventions apply to the $T$ variables and students in house $g$. Two sets are connected with a black line in Figure \ref{fig:alphaPart} if all students in one set must be connected to all students in the other set. These connections follow from the conditions of stability in Definition \ref{psDef}.  This gives us 3 constraints:
\begin{enumerate}
\item if $\alpha_\mu(s,g) = 1$ and $\alpha_\mu(t,h) = 1$ then $w(s,t) = 1$
\item if $\alpha_\mu(s,g) = 1$ and $\alpha_\mu(t,h) = 0$ then $w(s,t) = 1$
\item if $\alpha_\mu(s,g) = 0$ and $\alpha_\mu(t,h) = 1$ then $w(s,t) = 1$
\end{enumerate}

These constraints give us a lower bound on the edges between houses $h$ and $g$.
\begin{equation} \label{E_hg_spec_bound}
\sum_{t \in \mu(g)}\sum_{x \in \mu(h)}w(t,x) \geq S_1 T_1 + S_1 T_0 + S_0 T_1
\end{equation}

To prove the theorem, we want to find an upper bound on the cross edges, $E_{hg}$, so we relate the edges in the graph to the sum of the $\alpha$ values using the definition of the $\alpha$ function.
\begin{equation} \label{eq:alpha1}
\sum_{s \in \mu(h)}\alpha_\mu(s,g) =  q_h(D_g - D_h) + E_{hg} - 2E_{hh}
\end{equation}

Since the students in each house are partitioned by their $\alpha$ values, we can bound this sum as:
\begin{equation} \label{eq:alpha2}
\sum_{s \in \mu(h)}\alpha_\mu(s,g) \leq S_1 - S_{-1}
\end{equation}
Combining \eqref{eq:alpha1} and \eqref{eq:alpha2} gives
\begin{equation} \label{E_hg_inprogress}
E_{hg} \leq q_h(D_h - D_g) + 2E_{hh} + S_1 - S_{-1}
\end{equation}

To continue, we need to find an upper bound on the quantity $S_1 - S_{-1}$.  To do this, we start by lower bounding $E_{gg}$.
\begin{equation*}
2E_{gg} = \sum_{t \in \mu(g)}\sum_{x \in \mu(g)}w(t,x)
\end{equation*}
Recalling the definition of $\alpha$ in \eqref{alpha} gives
\begin{equation*}
\sum_{x \in \mu(g)}w(t,x) = D_h - D_g + \sum_{x \in \mu(h)}w(t,x) - \alpha_\mu(t,h).
\end{equation*}
Combining the previous two equations yields
\begin{align*}
2E_{gg} =& \sum_{t \in \mu(g)}\left( D_h - D_g + \sum_{x \in \mu(h)}w(t,x) - \alpha_\mu(t,h) \right) \\
=& q_g(D_h - D_g) + \sum_{t \in \mu(g)}\sum_{x \in \mu(h)}w(t,x) - \sum_{t \in \mu(g)}\alpha_\mu(t,h).
\end{align*}
Using inequalities \eqref{E_hg_spec_bound} and \eqref{eq:alpha2} gives
\begin{equation} \label{mid_ineq}
S_1 T_1 + S_1 T_0 + S_0 T_1 - (T_1 - T_{-1}) \leq 2E_{gg} + q_g (D_g - D_h)
\end{equation}

We can now use the above to find an upper bound on $S_1 - S_{-1}$.  To do this, we relate the left hand side of the above inequality to $S_1 - S_{-1}$.

Specifically, let $f(S_1,S_0,S_{-1},T_1,T_0,T_{-1}) = S_1 T_1 + S_1 T_0 + S_0 T_1 - T_1 + T_{-1} - (S_1 - S_{-1})$.  It is possible to show using elementary techniques that this function is non-negative, and thus that
\begin{equation} \label{S1Bound}
 S_1 - S_{-1} \leq S_1 T_1 + S_1 T_0 + S_0 T_1 - T_1 + T_{-1}
\end{equation}
We omit the details for brevity.  Note, however that the inequality in \eqref{S1Bound} holds only for the case where $q_h = q$ for all $h \in H$.  In the case where the quotas are not equal but $D_h = D$ for all $h \in H$, the proof technique differs slightly, but still yields $S_1 - S_{-1} \leq 2 E_{gg}$, from which the result follows.

Finally, combining (\ref{mid_ineq}) and (\ref{S1Bound}) gives
\begin{align*}
S_1 - S_{-1} & \leq S_1 T_1 + S_1 T_0 + S_0 T_1 - T_1 + T_{-1} \\
& \leq 2E_{gg} + q_g(D_g - D_h)
\end{align*}

To complete the proof we now plug the above into (\ref{E_hg_inprogress}) to get
\begin{align*}
E_{hg} \leq& q_h(D_h - D_g) + 2E_{hh} + 2E_{gg} + q_g(D_g - D_h) \\
\leq& \max (q_h(D_h - D_g), q_g(D_g - D_h)) + 2(E_{hh} + E_{gg}).
\end{align*}
where the final step follows from noting that at most one of $D_h-D_g$ and $D_g-D_h$ is strictly positive.
\end{proof}

%


\begin{lemma} \label{m-1 bound}
\begin{equation*}
\frac{\sum_{g<h \in H}q_h(D_h - D_g)}{\sum_{h \in H}q_h D_h} \leq m-1
\end{equation*}
\end{lemma}

\begin{proof}
Without loss of generality assume the houses are ordered so that if $g < h$, then $D_g < D_h$.  The following inequalities hold simply because $q_h, q_g, D_h, D_g$ are all non-negative values.
\begin{align*}
\frac{\sum_{g<h \in H}q_h(D_h - D_g)}{\sum_{h \in H}q_h D_h} \leq& \; \frac{\sum_{g<h \in H} (q_h D_h + q_g D_g)}{\sum_{h \in H}q_h D_h} \\
\leq& \frac{\sum_{h \in H} \sum_{g \neq h \in H} q_h D_h}{\sum_{h \in H}q_h D_h} \\
=& \; \frac{\sum_{h \in H} (m-1) q_h D_h}{\sum_{h \in H}q_h D_h} \\
=& \; m-1
\end{align*}
\end{proof}

%

The remaining lemmas parallel the above lemmas, but are used for proving Theorem \ref{thm:PoAGeneral}, and thus apply in more general settings.

\begin{lemma} \label{E_hg_1_general}
Let $w(s,t) \in \mathbb{R}^+\cup\{0\}$ for all students $s,t$, and let $D_h \in \mathbb{R}^+\cup\{0\}$ for all $h \in H$.
Consider a stable matching $\mu$.  If there exists an  $s \in \mu(h)$ such that $\alpha_\mu(s,g) > w(s,t)$ for all $t \in \mu(g)$, then
$E_{gh} < q_g(D_g-D_h) + 2 E_{gg} + q_g w_{max}.$
\end{lemma}

\begin{proof}
By assumption, there exists a student in $h$ that strictly wants to swap with any student in house $g$.  It then follows from the stability of $\mu$ that \emph{all} students in $g$ must strictly oppose the swap (i.e., $\alpha_\mu(t,h) < w(s,t)$).  This gives
\begin{align*} D_h - D_g+  \sum_{x \in \mu(h)} w(t,x) - \sum_{x \in \mu(g)} w(t,x) < \; w(s,t)
< \; w_{max},
\end{align*}
for all $t \in \mu(g).$
Summing the above equation over $t \in \mu(g)$ then yields
\begin{equation*}
q_g(D_h - D_g) + E_{gh} - 2E_{gg} < q_g w_{max}
\end{equation*}
Rearranging the previous equation completes the proof.
\end{proof}


\begin{lemma} \label{E_hg_3_general}
Let $w(s,t) \in \mathbb{R}^+\cup\{0\}$ for all students $s,t$, and let $D_h \in \mathbb{R}^+\cup\{0\}$ for all $h \in H$.
Consider a stable matching $\mu$.   If (i) there does not exist an $s \in \mu(h)$ such that $\alpha_\mu(s,g) > w(s,t)$ for all $t \in \mu(g)$ and (ii) there does not exist $t \in \mu(g)$ such that $\alpha_\mu(t,h) > w(s,t)$ for all $s \in \mu(h)$, then
\begin{align*} E_{hg} \leq & \max(q_h(D_h-D_g),q_g(D_g-D_h)) \\
& + 2 (E_{hh} + E_{gg}) + q_{max} w_{max} \end{align*}
\end{lemma}
\begin{proof}
Conditions (i) and (ii) are equivalent to requiring
\begin{align*}
\forall s \in \mu(h), & \; D_g - D_h + \sum_{x \in \mu(g)} w(s,x) \\
& - \sum_{x \in \mu(h)} w(s,x) \leq w(s,t) \; \forall t \in \mu(g)
\end{align*}
and
\begin{align*}
\forall t \in \mu(g), & \; D_h - D_g + \sum_{x \in \mu(h)} w(t,x) \\
& - \sum_{x \in \mu(g)} w(t,x) \leq w(s,t) \; \forall s \in \mu(h)
\end{align*}
To complete the proof we simply sum these two bounds using $w(s,t) \leq w_{max}$ and $E_{gh} = E_{hg}$.
\end{proof} 

\end{document}